\documentclass[aps,prx,reprint,showpacs,superscriptaddress,nofootinbib]{revtex4-2}
\usepackage[utf8]{inputenc}
\usepackage[T1]{fontenc}
\usepackage{amsmath, amsthm, amssymb, mathtools, dsfont} 
\usepackage{fixmath}
\usepackage{times}
\usepackage{graphicx}  
\usepackage[all]{xy}
\usepackage{tikz}
\usepackage{adjustbox}
\usetikzlibrary{positioning, shapes,arrows.meta}
\usepackage{ragged2e}
\usepackage{dcolumn} 
\usepackage{bm,bbm} 
\usepackage[normalem]{ulem}

\usepackage{esint}
\usepackage{paralist}
\usepackage[normalem]{ulem}

\usepackage{thmtools}
\usepackage{thm-restate} 

\usepackage{verbatim} 

\usepackage{geometry}
\definecolor{myurlcolor}{rgb}{0,0,0.7}
\definecolor{myrefcolor}{rgb}{0.8,0,0}
\usepackage[unicode=true,pdfusetitle,  bookmarks=true,bookmarksnumbered=false,bookmarksopen=false, breaklinks=false,pdfborder={0 0 0},backref=false,colorlinks=true, linkcolor=myrefcolor,citecolor=myurlcolor,urlcolor=myurlcolor]{hyperref}
\DeclareGraphicsExtensions{.png,.pdf,.eps}
\graphicspath{ {./figures/} }
\usepackage{epstopdf}
\usepackage[capitalize,nameinlink,noabbrev]{cleveref}

\usepackage{mathtools}
\DeclarePairedDelimiter{\ceil}{\lceil}{\rceil}

\makeatletter

 \theoremstyle{plain}
 
 \theoremstyle{plain}
 \newtheorem{lem}{Lemma}
 \theoremstyle{plain}
 \newtheorem{thm}{Theorem}
 \theoremstyle{plain}
  \newtheorem{res}{Result}
 \theoremstyle{plain}
   
 \theoremstyle{plain}
  \newtheorem{prop}{Proposition}
 \theoremstyle{plain}
 
 \theoremstyle{plain}
 \newtheorem{corr}{Corollary}
 \theoremstyle{plain}
 
 \theoremstyle{remark}
 \newtheorem*{rem*}{Remark}
 \theoremstyle{plain}
  \newtheorem{rem}{Remark}

\theoremstyle{plain}
 \newtheorem*{conj*}{Conjecture}
 \theoremstyle{plain}

\newcommand{\ot}{\otimes}

\renewcommand{\exp}{\mathrm{exp}}
\DeclareMathOperator{\tr}{tr}

\renewcommand{\H}{\mathcal{H}}

\newcommand{\C}{\mathbb{C}} 
\newcommand{\E}{\mathbb{E}} 
\newcommand{\M}{\mathbf{M}} 
\newcommand{\N}{\mathbf{N}} 
\newcommand{\F}{\mathbf{F}} 
\renewcommand{\L}{\mathbf{L}} 
\renewcommand{\P}{\mathbf{P}} 
\newcommand{\Pp}{\mathbb{P}} 
 
\newcommand{\Q}{\mathcal{Q}} 

\newcommand{\PP}{\mathbb{P}} 

\newcommand{\SP}{\mathrm{S}\mathbb{P}} 
\newcommand{\POVM}{\mathrm{POVM}} 

\newcommand{\Herm}{\mathrm{Herm}} 
\newcommand{\I}{\mathbb{I}} 

\newcommand{\norm}[1]{\left|\left|  \ #1  \ \right| \right|}

\newcommand{\bo}{n}

\newcommand{\myvdots}{\vbox{\baselineskip4\p@ \lineskiplimit\z@
    \hbox{.}\hbox{.}\hbox{.}}} 

\newcommand{\mk}[1]{{\color{blue} #1}}
\newcommand{\mo}[1]{{\color{red} #1}}

\global\long\global\long\global\long\def\bra#1{\mbox{\ensuremath{\langle#1|}}}
\global\long\global\long\global\long\def\ket#1{\mbox{\ensuremath{|#1\rangle}}}

\makeatother
\renewcommand{\ket}[1]{\left| #1 \right>} 
\renewcommand{\bra}[1]{\left< #1 \right|} 
\newcommand{\ketbra}[2]{\left| #1 \rangle\langle #2 \right|} 

\usepackage[framemethod=TikZ]{mdframed}
\definecolor{thmgray}{gray}{0.95}
\definecolor{conclusionbackground}{gray}{0.95}
\newcommand{\thmspaceafter}{2mm}

{\begin{mdframed}[backgroundcolor=thmgray,nobreak,roundcorner=5pt,linewidth=1pt]\begin{myresult}}%
		{\end{myresult} \vspace{\thmspaceafter} \end{mdframed}}

\begin{document} 	

\title{Pretty-good simulation of all quantum measurements by projective measurements}
	 
	\author{Micha\l\  Kotowski}

    \affiliation{Institute of Mathematics, Faculty of Mathematics, Informatics, and Mechanics, University of Warsaw, Banacha 2, 02-097 Warsaw, Poland} 

	\author{Micha\l\ Oszmaniec}
	\email{oszmaniec@cft.edu.pl}
	\affiliation{Center for Theoretical Physics, Polish Academy of Sciences, Al. Lotnik\'ow 32/46, 02-668 Warsaw, Poland}
 
 \

\begin{abstract}
In quantum theory general measurements  are described by so-called Positive Operator-Valued Measures (POVMs). We show that in $d$-dimensional quantum systems an application of depolarizing noise with constant (independent of $d$) visibility parameter makes \emph{any} POVM simulable by a randomized implementation of projective measurements that do not require any auxiliary systems to be realized. 

This result significantly limits the asymptotic advantage that POVMs can offer over projective measurements in various information-processing tasks, including state discrimination, shadow tomography or quantum metrology. We also apply our findings to  questions originating from quantum foundations by asymptotically improving the range of visibilities for which noisy pure states of two qudits admit a local model for generalized measurements. As a byproduct, we give asymptotically tight (in terms of dimension) bounds on critical visibility for which all POVMs are jointly measurable. 

On the technical side we use recent advances in POVM simulation, the solution to the celebrated Kadison-Singer problem, and a method of approximate implementation of ``nearly projective'' POVMs by a convex combination of projective measurements, which we call dimension-deficient Naimark theorem. Finally, some of our intermediate results show (on information-theoretic grounds) the existence of circuit-knitting strategies allowing to simulate general $2N$ qubit circuits by randomization  of subcircuits operating on $N+1$ qubit systems, with a constant (independent of $N$) probabilistic overhead.
\end{abstract}
\maketitle	
 
\section{Introduction}
In quantum mechanics, contrary to classical physics, the act of measurement plays a  prominent role. While in the classical picture of the world physical objects have well-defined attributes that are merely revealed by performing a measurement, in quantum physics system's characteristics can be viewed as emerging in the course of the measurement process itself. In quantum theory general measurement procedures that can be performed on a physical system are described by mathematical objects called Positive Operator-Valued Measures (POVMs) \cite{Peres2002}. The most commonly encountered POVMs are called projective or von Neumann measurements and are realized by measuring observables on a given system (such as its energy, angular momentum etc.). To  physically realize a general POVM, it is often necessary to extend the system of interest by an ancilla and then perform a projective measurement on the combined system, which renders general POVMs much more difficult to implement compared to projective measurements \cite{Oszmaniec17}. 

Generalized measurements find many applications across quantum information and quantum computing: in study of nonlocality \cite{Barret2002,Vertesi2010}, entanglement detection \cite{Shang2018}, randomness generation \cite{Acin2016}, discrimination of quantum states \cite{QuantumStateDisriminationRev},  (multi-parameter) quantum metrology protocols \cite{Ragy2016,Szcyykulska2016,rafalMETRO}, attacks on quantum cryptography \cite{optimalCRYPTO}, its shadow version \cite{singleSettingTom,povmSHADOWGhune}, quantum tomography \cite{Derka1998,Renes2004,OptTomography}, quantum algorithms \cite{algREVIEW,Bacon2006,HSP} or port-based teleportation  \cite{Ishizaka2008,Studzinski2017port,Mozrzymas2018optimal}, to name just a few. At the same time, the relative usefulness and \emph{advantage} that POVMs can offer over projective measurements for different tasks with increasing Hilbert space dimension remain poorly understood (despite some partial results in that direction based on resource-theoretic approaches \cite{OszmaniecBiswas2019, Uola2019, ResTheoryMeas,Buscemi2024completeoperational} and specific simulation strategies \cite{Oszmaniec17,Oszmaniec19,SMO2022}). 

In this work we show that a surprisingly broad class of POVMs in $d$-dimensional quantum systems can be simulated by a randomization of measurements that do not require auxiliary systems or need only a single qubit to be implemented (see Fig \ref{fig:diagram} for a graphical presentation of the results). Specifically, for a qudit POVM $\M$ we analyze the action of the depolarizing channel\footnote{For an arbitrary POVM $\M$, its depolarized version $\mathrm{\Phi}_t(\M)$ describes a measurement in which a quantum measured state is first affected by the depolarizing channel $\mathrm{\Phi}_t$ and then the POVM $\M$ is implemented.} $\mathrm{\Phi}_t(\M)$ and show that for $c=0.02$ (a dimension-independent constant) $\mathrm{\Phi}_c(\M)$ can be realized by randomization of projective measurements. We furthermore show a related result -- an arbitrary qudit POVM $\M$ can be simulated with postselection probability $q=1/8=0.125$ (again, a dimension -independent constant) by a convex combination of measurements requiring only a single auxiliary qubit to be implemented. Here, simulation with postselection refers to the protocol that, in every experimental round, realizes the target measurement with probability $q$ or reports failure with probability  $1-q$. These results limit the asymptotic advantage that general POVMs  can offer over projective measurements or measurements requiring small ancillas, as long as they can be complemented with classical randomness and post-processing.

\begin{figure*}
\centering
\begin{adjustbox}{width=.9\textwidth}

\includegraphics{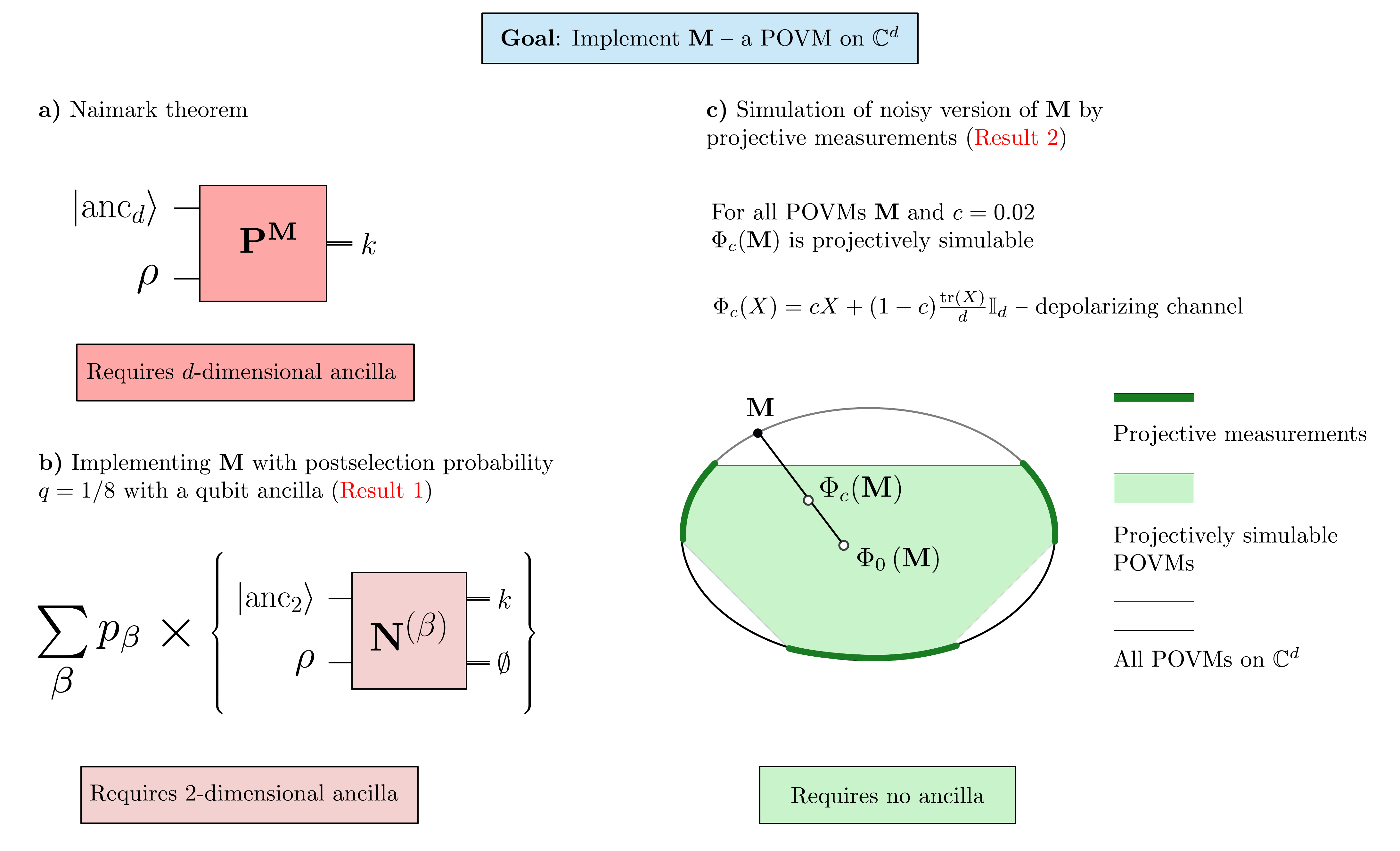}
\end{adjustbox}

\caption{\label{fig:diagram} We tackle the problem of minimizing the size of the ancilla space necessary to implement a general POVM $\M$ on a $d$-dimensional system. (a) The standard method is based on the Naimark theorem (Theorem \ref{th:Naimark}) which, in conjunction with convex structure of POVMs (cf. \cite{Oszmaniec17}), allows to implement arbitrary $\mathbf{M}$ by projective measurement $\P^\M$ on a system enlarged by a $d$-dimensional ancilla initialized in state $\ket{\mathrm{anc}_d}$. (b) An alternative solution, formalized as Result \ref{res:singleANC}, allows to realize arbitrary qudit POVM via randomization over 
 measurements $\N^{(\beta)}$ which require only a single auxiliary qubit (initialized in state $\ket{\mathrm{anc}_2}$) to be implemented. This simplification comes at a price -- the protocol works in every experimental shot with success probability $q=1/8$. (c) If no ancillas are permitted it is possible to realize a noisy version of $\M$ by a convex combination of projective measurements on $\C^d$. Specifically, in Result \ref{res:projSIM} we prove that for $c=0.02$ the noisy POVM $\mathrm{\Phi}_c(\M)$ is projectively simulable, where $\mathrm{\Phi}_c$ is the depolarizing channel.
 }
\end{figure*}

Our findings have significant consequences for various areas of quantum information and quantum computing. First, they allow us to develop a better hidden variable model for entangled  noisy pure states of two qudits (improving over previously existing results \cite{Almeida2007}).  As a consequence we get asymptotically tight robustness bounds for incompatibility of POVMs in a $d$-dimensional space (improving over state-of-the-art results from \cite{Almeida2007,Wiseman2007}). Furthermore, on the applied side, our results limit asymptotic usefulness of generalized measurements (compared to projective measurements) for a variety of tasks, including quantum state discrimination \cite{QuantumStateDisriminationRev}, shadow tomography \cite{ReviewShadows23} or multiparameter quantum metrology \cite{Szcyykulska2016}. Finally, our findings inspire a new circuit knitting scheme which is conceptually different than previous proposals (see e.g. \cite{knitting0,knitting1,knitting2}) and allows (on information-theoretic grounds) to simulate general $2N$ qubit circuits by a convex combination of subcircuits operating on $N+1$ qubits, with a constant (independent of $N$) success probability $q=1/8$.  

On the technical side, our work relies on a recent POVM simulation protocol developed in \cite{SMO2022}, a new technique for approximate implementation of \emph{nearly projective} measurements by a convex combination of projective measurements (which we call dimension-deficient Naimark theorem), and uses the solution of the celebrated Kadison-Singer problem. This problem, originally posed in \cite{kadison-singer} in the context of operator algebras and foundations of quantum theory, could be informally stated as follows \cite{marcus2016solution}: given a quantum system, does knowing the outcomes of all measurements with respect to a maximal set of commuting observables uniquely determine the outcomes of all possible measurements of all possible observables? In mathematical terms, this corresponds to the question of whether every pure state on a maximal algebra of bounded diagonal operators has a unique extension to the algebra of all bounded operators. The problem has gained further prominence in other fields of mathematics such as operator theory, harmonic analysis or frame theory \cite{dolbeault-cohen}. At the same time, it had resisted solution for multiple decades, until the recent and unexpected breakthrough of Marcus, Spielman and Srivastava (\cite{marcus-spielman-srivastava}), solving the conjecture in the affirmative. We refer the reader to \cite{casazza-tremain} for a more thorough discussion of the mathematical context of this result and to \cite{xu-xu-zhu} for some of the more recent improvements and follow-up work. Interestingly, the solution of the Kadison-Singer conjecture can be phrased as a specific statement about POVMs on finite dimensional spaces, which is crucial to our analysis.

\emph{Relation to previous works.} The problem of physical realization of generalized measurements and their simulations received significant attention in recent years. In certain scenarios,  it is possible to implement POVMs by using the standard Naimark recipe (see Theorem \ref{th:Naimark} below), utilizing additional degrees of freedom (like extra modes in optical systems \cite{OptExp23}, additional energy levels present in trapped ions \cite{singleSettingTom} and anharmonic oscillators in superconducting systems \cite{ancFREE}). Such an approach, however, comes with additional experimental cost, and may be difficult to realize in scenarios when we want to implement nontrivial measurement acting on many qubits at the same time (in general, a $d$-dimensional ancilla is needed to implement an arbitrary POVM on a $d$-dimensional system \cite{Oszmaniec17}). Another known strategy of realizing POVMs involves performing a sequence of adaptive non-destructive\footnote{In this context non-destructive refers to the property that the state is not destroyed in the course of the measurement and further operations can be carried out on it after the measurement result is obtained. } binary measurements organized in a \emph{search tree} tailored to the target measurement \cite{TreePOVM}. This method utilizes a single auxiliary qubit. Yet, adaptive measurements introduce additional errors which build up with the number of adaptive steps in the algorithm. To remedy this, a recent work \cite{ExpTreeIBM} introduced a hybrid scheme which combined Naimark dilation and the \emph{search tree} approach. 

Our results build on a line of work which aims to realize POVMs without additional quantum resources such as large ancillas or adaptive measurements. 
In \cite{Oszmaniec17} a set of projectively simulable measurements $\SP(d)$  was defined as a class of measurements on $\C^d$  that can be implemented by using  randomization and postselection of projective measurements on $\C^d$ (see also the concurrent work \cite{Hirsch2017betterlocalhidden} where the concept of simulation of POVMs via projective measurements was used to construct new local models for Werner states on two qubits). In the follow-up works, general simulability properties of POVMs (with respect to classes of measurements different than projective) were discussed in \cite{Guerini2017}, and in \cite{Oszmaniec19} an optimal simulation strategy for simulation of general POVMs with projective measurements \emph{and postselection} was proposed, showing that for this strong notion of simulation for some POVMs the success probability cannot exceed $q=1/d$. Subsequently, a general resource-theoretic perspective on this and related problems was provided in a series of papers \cite{Skrzypczyk24,OszmaniecBiswas2019,ResTheoryMeas,Uola2019} that connected the problem of simulation of POVMs  via measurements from a free convex set $\mathcal{F}$ to finding maximal advantage that a given measurement can offer over free measurements for state discrimination \cite{QuantumStateDisriminationRev}. 
 In a recent paper \cite{SMO2022} a generalization of the scheme from \cite{Oszmaniec19} was proposed, which proved surprisingly effective in simulating  general quantum measurements on $\C^d$ via  measurements requiring only a single auxiliary qubit to be implemented. Specifically, it was shown that the method is capable of simulating Haar-random measurements on $\C^d$ with constant success probability, but there was no rigorous argument for the performance of the protocol for general POVMs (the reason being a complicated combinatorial optimization problem whose solution enters the definition of the simulation protocol -- see discussion below Theorem \ref{th:OldProtocol}). One of the contributions of the present work is the resolution of this issue by utilizing the celebrated solution to the  Kadison-Singer conjecture due to Marcus-Spielman-Srivastava  \cite{marcus-spielman-srivastava}. 

Finally, we remark that there also exists another method of  approximate simulation of POVMs based on sparsification \cite{Lancien}, aiming to minimize the number of outcomes while keeping a similar \emph{so-called} distinguishability norm of a POVM. Despite being conceptually related, our results concern an entirely different problem than the one considered in \cite{Lancien}.

\emph{Organization of work.} In Section \ref{sec:preliminaries} we survey basic concepts and notation on POVMs that will be used throughout the paper. Then, in Section \ref{sec:mainRES} we first state our key results  about POVM simulation, and then discuss their applications in various areas of quantum information and computing, concluding with an outlook for further work. The rest of the paper is devoted to formally proving our results. First, in  Section \ref{sec:overview} we present a detailed outline of the proof of our key results. The subsequent Section \ref{sec:Postselection} contains rigorous presentation of probabilistic simulation of POVMs with strategies using ancillas of bounded dimension. In Section \ref{sec:defitientNAIMARK} we state and prove dimension-deficient Naimark theorem  which underpins our result about POVM simulation strategies that do not use any ancilla. The main text is complemented by three appendices.  Appendix \ref{app:aux} contains auxilary technical results. In Appendix \ref{app:Applications} we give detailed computations relevant for applications discussed earlier in Section \ref{sec:mainRES}. Finally, in Appendix \ref{app:randomPART} we outline a computationally efficient (in $d$) method for simulating general POVMs, albeit with smaller success probability.

\section{Preliminaries}\label{sec:preliminaries}

We start by providing the necessary background and notation on POVMs, their simulability properties, and mathematics that will be used in the rest of the paper.

We will be concerned with generalized measurements on a $d$-dimensional Hilbert space $\H\simeq \C^d$.  A positive operator-valued measure (POVM) is a tuple $\M = \left(M_1, M_2, \ldots, M_\bo \right)$ of non-negative ($M_i\geq 0$) operators (called effects)  normalized to identity ($\sum_{i=1}^\bo M_i = \I_d$). We will denote the set of all quantum measurements on a finite dimensional Hilbert space $\H$  by $\POVM(\H)$. When  $\M$ is performed on a quantum state $\rho$, it yields a random outcome $i$, distributed according to the probability distribution $\mathrm{p}(i|\rho, \M) = \tr\left( \rho M_i \right)$ (Born rule). A measurement $\mathbf{P}=\left(P_1,P_2,\ldots,P_n\right)$ is called projective if its effects satisfy $P_i P_j =\delta_{ij} P_i$. The set of projective measurements on $\H$ will be denoted by $\PP(\H)$. According to the postulates of quantum mechanics, a projective measurement can be associated to measurements of a quantum observable $O$ (effects $P_i$ are projections onto eigenspaces of $O$).

\begin{rem}\label{rem:OutcomeNumber}
    
Throughout the paper we will be considering measurements with \emph{finite} number of outcomes. However, our results naturally extend to the general POVMs on $\C^d$ with continuous output sets. This is because of the main result of \cite{Chiribella2007} which shows that general continuous-outcome measurements on $\C^d$ can be realized as randomization of measurements with at most $d^2$ nonzero effects.  We are not concerned with the state of the system after the measurement -- incorporating this would require considering  more complicated objects: so-called quantum instruments \cite{heinosaari_ziman_2011}. 
\end{rem}

A natural way to physically realize a generalized quantum measurement is to perform a suitable projective measurement on an extended space.

\begin{thm}[Naimark extension theorem \cite{Peres2002}]\label{th:Naimark}
    Let $\M=(M_1,\ldots,M_n)$ be an $n$-outcome POVM on $\H\simeq\C^d$ with rank-one effects (i.e. $M_i=\alpha_i \ketbra{\psi_i}{\psi_i} $). Then there exists a Hilbert space $\H_{ex}\simeq \C^n$, which contains $\H$ as a subspace, and a rank-one projective measurement $\P^\M=(\ketbra{\phi_1}{\phi_1},\ldots, \ketbra{\phi_n}{\phi_n})$ on $\H_{ex}$,  such that for all states $\rho$ supported on $\H$  and all outcomes $i$ we have
    \begin{equation}\label{eq:NaimarkCondition}
        \tr(\rho M_i) =\tr(\rho \ketbra{\phi_i}{\phi_i})\ .
    \end{equation}
\end{thm}
The extended space is often realized as  a tensor product with an ancilla space, i.e. $\H_{ex}=\H\ot \H_{a}$. In order to realize the embedding of $\H$ into $\H_{ex}$, one typically fixes a state $\ket{\psi_0}\in \H_a$ and identifies the subspace $\tilde{\H}=\mathrm{span}_\C \lbrace\ket{\psi}\ket{\psi_0}, \ket{\psi}\in \H\rbrace$ with $\H$.  The drawback of the Naimark theorem is that the dimension of the extended space $\H_{ex}$ (and thus also the dimension of the ancilla space $\H_a$) is proportional to the number of outcomes $n$ of the target POVM $\M$. 

It is however possible to reduce the dimension of the ancilla space by complementing quantum resources with suitably chosen classical processing of generalized measurements. There exist three natural classical operations that can be applied to POVMs: taking convex combinations, post-processing and postselection. A convex combination of two POVMs $\M,\N$ (with the same number of outcomes $n$) can be operationally interpreted as a POVM realized by applying, in a given experimental run, measurements $\M$ and $\N$ with probabilities $p$ and $1-p$ respectively (where $p\in[0,1]$ is the mixing parameter). The resulting POVM  is denoted by $p \M + (1-p) \N$ and its $i$-th effect is given by $\left[ p \M + (1-p) \N \right]_i = p M_i + (1-p) N_i$.  See \cite{DAriano2005} for a thorough exposition of the convex structure of quantum measurements. Classical post-processing, on the other hand, refers to application of probabilistic relabelling of the measurements outcomes \cite{Buscemi2005,Haapasalo2012}. For an $n$-outcome POVM $\M$, the application of post-processing results in another POVM $\Q(\M)$ with $n'$ outcomes and effects $\Q(\M)_i = \sum_{j} q_{i|j} M_j$, where $q_{i|j}$ is an $n'\times n$ array of conditional probabilities, i.e.,  $q_{i|j}\ge 0$ and $\sum_{i} q_{i|j} =1$ for each $j$. A class of projectively simulable measurements, introduced in \cite{Oszmaniec17} and denoted by  $\SP(d)$, comprises measurements on $\C^d$ that can be realized by randomization and post-processing of projective measurements on $\C^d$, i.e., without the use of ancillary systems. 

Lastly, postselection, i.e., the process of disregarding certain outcomes, can be used to  implement otherwise inaccessible POVMs. 
We say that a POVM $\L=(L_1,\ldots,L_\bo,L_{\bo+1})$ simulates a POVM $\M=(M_1,\ldots,M_\bo)$ with postselection probability $q$ if $L_i=q M_i$ for $i=1,\ldots,\bo$. When we  implement $\L$, then, conditioned on getting the first $\bo$ outcomes, we obtain samples from $\M$.  Thus, we can simulate $\M$ by implementing $\L$ and postselecting on non-observing the outcome $\bo+1$.
The probability of successfully doing so is $q$, which means that a single sample of $\M$ is obtained by implementing $\L$  on average $1/q$ number of times. The reader is referred to \cite{Oszmaniec19} for a more detailed discussion of simulation via postselection.  

Throughout this work we will be extensively using depolarizing channel acting on quantum states and (by duality) on quantum measurements. For $t\in[0,1]$ (known as the visibility parameter) and $X\in\Herm(\C^d)$ we define $\mathrm{\Phi}_t(X)\coloneqq tX+(1-t)\frac{\tr(X)}{d}\I_d$, where $\I_d$ is identity operator on $\C^d$. Action of depolarizing channel on a quantum measurement $\M=(M_1,\ldots,M_n)$ is defined by setting  $\mathrm{\Phi}_t(\M)$ to be a measurement whose effects are depolarized versions of the original measurement operators: $[\mathrm{\Phi}_t(\M)]_i= \mathrm{\Phi}_t(M_i)$. Note that due to self-duality of $\mathrm{\Phi}_t$ (with respect to Hilbert-Schmidt inner product) and functional form of the Born rule we have  $\mathrm{p}(i|\mathrm{\Phi}_t(\rho), \M)= \mathrm{p}(i|\rho, \mathrm{\Phi}_t(\M))$, i.e. output statistics of a depolarized state $\mathrm{\Phi}_t(\rho)$ measured by the ideal measurement $\M$ are identical to measurement statistics of a noisy POVM $\mathrm{\Phi}_t(\M)$ on the ideal state $\rho$. Following \cite{Oszmaniec17} we will use white noise robustness to quantify the non-projective character of a POVM on $\C^d$:
\begin{equation}\label{eq:white_noise_robustness}
    t_{\SP}(\M)= \max \left\{ t\ |\ \mathrm{\Phi}_t(\M)\in \SP(d) \right\}\ .
\end{equation}
In a given dimension $d$, the minimal value of this function quantifies the worst-case robustness of quantum measurements against projective simulability: 
\begin{equation}\label{eq:worst-case-visibility}
t_{\SP}(d) = \min_{\M\in\POVM(\C^d)} t_{\SP}(\M).
\end{equation}
Prior to this work it was known that \cite{Oszmaniec17} $t_{\SP}(d)\geq \frac{1}{d}$. However, no matching \emph{upper bounds} on  this quantity were known. 

\textbf{Notation. } We will use $\|A\|$ to denote the operator norm of a linear operator $A$, and $n$ to denote $n$-element set $\lbrace1,\ldots\bo\rbrace$. Finally, for two positive-valued functions $f(x),g(x)$ we will write $f=\Theta(g)$ if there exist positive constants $c,C>0$ such that $c f(x)\leq g(x)\leq C f(x)$ for sufficiently large $x$.

\section{Main results and consequences}\label{sec:mainRES}

In this section we present our main findings regarding the simulation of general POVMs via simpler classes of measurements and classical resources. Our first finding confirms a conjecture from \cite{SMO2022} which asserted that arbitrary POVMs can be simulated with constant success probability by measurements requiring only a single auxiliary qubit to be implemented. 

\begin{mdframed}[backgroundcolor=thmgray,nobreak,roundcorner=5pt,linewidth=1pt]
\begin{res}[Constant success probability of simulation via POVMs with a single auxiliary qubit] \label{res:singleANC}Let $\M$ be a POVM on $\C^d$. Then there exists a probability distribution $\{p_\beta\}$ and a collection of POVMs $\{\N^{(\beta)}\}$ such that (i) for every $\beta$ the POVM $\N^{(\beta)}$ has at most $d+1$ outcomes and can be implemented by a single projective measurement on $\C^d \ot \C^2$ (i.e. using only a single auxiliary qubit as an ancilla), followed by classical post-processing, (ii) the convex combination $\L=\sum_{\beta}p_\beta \N^{(\beta)}$ simulates $\M$ with constant postselection probability $q=1/8=0.125$, i.e. $\sum_{\beta}p_\beta \N^{(\beta)}=(q\M,(1-q)\I_d)$. 
\end{res}
\end{mdframed}

\noindent The proof of this result is provided in Section \ref{sec:Postselection}. Therein, we also give a general version of this result (Theorem \ref{th:ancillas}), which shows that by allowing to use POVMs that employ $k$-dimensional ancillas, it is possible to simulate an arbitrary measurement with success probability $q\geq1-\Theta(k^{-1/2})$. 

Our second result concerns the scenario when we are not allowed to use any auxiliary qubits. In this case  it is impossible to implement an arbitrary POVM on $\C^d$ with success probability greater than $1/d$ (cf.  \cite{Oszmaniec19}). For this reason we turn to simulation of noisy (depolarized) versions of quantum measurements.   


\begin{mdframed}[backgroundcolor=thmgray,nobreak,roundcorner=5pt,linewidth=1pt]
\begin{res}[Depolarizing noise with constant visibility makes arbitrary POVM simulable by projective measurements] \label{res:projSIM}  Let $\M=(M_1,\ldots, M_n)$ be a POVM on $\C^d$. Then, for $c=0.02$ we have $\mathrm{\mathrm{\Phi}}_c(\M)\in\SP(\C^d)$. In other words, for every POVM $\M$  its noisy version with effects $\left[\mathrm{\Phi}_c(\M)\right]_i=c M_i +(1-c)\frac{\tr(M_i)}{d} \I_d $ can be simulated as a convex combination of projective measurements that do not require any ancillas to be implemented.
\end{res}
\end{mdframed}

The crucial feature of both Result \ref{res:singleANC} and  Result \ref{res:projSIM} is that the parameters $q,c$ are independent of the dimension of the Hilbert space.  This is in stark contrast to a number of prior results regarding white noise robustness of entanglement \cite{robustnessEntanglement}, nonlocality \cite{robustnessNONLOCALITY}, non-Gaussianity for fermionic systems \cite{fermRobustness} or incompatibility of projective measurements \cite{JMrobustness}.

Although simple to state, our findings give rise to a number of nontrivial consequences in quantum information and quantum computing, which we describe in what follows. We will focus on explaining the context and relevance of different applications, delegating more involved proofs of technical statements to Appendix \ref{app:Applications}.

\subsection{Limitation of usefulness of POVMs in state discrimination and other linear games}\label{sub:statediscrimination}

Quantum state discrimination is one of the primary
applications of POVMs. There exist many variants of this problem (see \cite{QuantumStateDisriminationRev} for a recent review) but in what follows we will focus on its most basic incarnation, the so-called minimal-error state discrimination. In this scenario one is given a source  of quantum states  which generates a state $\rho_i$ with a priori probability $p_i$ (the source can be characterized by an ensemble of quantum states $\mathcal{E}=\{p_i,\rho_i\}_{i=1}^n$). The task is then to find the label $i$ by performing a POVM $\M = (M_1,\ldots,M_n)$ on an unknown quantum state generated by $\mathcal{E}$. In the problem of minimal-error state discrimination, one is interested in optimizing the average success probability $p_\mathrm{succ}(\mathcal{E},\M)=\sum_{i=1}^n p_i \tr(\rho_i M_i)$. Besides foundational interests, minimal-error state discrimination appears naturally in different contexts, as many tasks can be phrased as a variant of this problem: quantum communication \cite{minMAXentropy2009}, asymptotic quantum cloning \cite{cloningBae}, finding hidden subgroup states \cite{Bacon2006,HSP}, or port-based teleportation \cite{Ishizaka2008,Studzinski2017port,Mozrzymas2018optimal}. The following proposition shows the limitation of the relative power of generalized measurements over projective measurements (as a function of dimension $d$):
\begin{prop}
[No unbounded advantage of general POVMs over projective measurements in minimal-error state discrimination]\label{prop:state discrimination} Let $\mathcal{E}=\{p_i,\sigma_i\}_{i=1}^n$ be a source of quantum states on $\C^d$. Let $p_\mathrm{succ}(\mathcal{E},\M)=\sum_{i=1}^n p_i \tr(\sigma_i M_i)$ be the success probability of discriminating states in $\mathcal{E}$ via POVM $\M=(M_1,\ldots,M_n)$. Then, for every  $\mathcal{E}$ and $\M$, a projectively simulable POVM $\N=\mathrm{\Phi}_c(\M)\in\SP(d)$ satisfies $c\,p_\mathrm{succ}(\mathcal{E},\M) \leq \ p_\mathrm{succ}(\mathcal{E},\N)$, where $c=0.02$. Consequently, we have 
\begin{equation}\label{eq:statediscrADV}
(1/c)\max_{\P\in\PP(\C^d)} p_\mathrm{succ}(\mathcal{E},\P) \geq \max_{\M\in\POVM(\C^d)} p_\mathrm{succ}(\mathcal{E},\M)\ .
\end{equation}
\end{prop}
\begin{proof}
 We have  
\begin{equation}
p_\mathrm{succ}(\mathcal{E},\mathrm{\Phi}_c(\M))= c\  p_\mathrm{succ}(\mathcal{E},\M) + (1-c) p_\mathrm{succ}(\mathcal{E},\mathrm{\Phi}_0(\M))\ 
\end{equation}
and additionally $ p_\mathrm{succ}(\mathcal{E},\mathrm{\Phi}_0(\M)) \geq 0$. By realizing that $p_\mathrm{succ}(\mathcal{E},\cdot)$ is linear in the second argument we get
\begin{equation}
    \max_{\P\in\PP(\C^d)} p_\mathrm{succ}(\mathcal{E},\P) \geq p_\mathrm{succ}(\mathcal{E},\mathrm{\Phi}_c(\M))\ \geq c\  p_\mathrm{succ}(\mathcal{E},\M)\ .
\end{equation}
We obtain Eq. \eqref{eq:statediscrADV} by  optimizing over $\M\in\POVM(\C^d)$.

\end{proof}
The papers \cite{OszmaniecBiswas2019,Uola2019,Takagi2019} established a quantitative connection between minimal-error state discrimination and a geometric measure of resourcefulness of POVMs with respect to compact subsets $\mathcal{F}$ of POVMs $\POVM(\C^d)$, called generalized robusteness:
\begin{equation}
    R_\mathcal{F}(\M)=\min \left\lbrace s \left| \exists \N \in \POVM(\C^d)\ \text{s.t}\ \frac{\M+s \N}{1+s}\in\mathcal{F}  \right. \right\rbrace\ .
\end{equation}
Specifically, we have
\begin{equation}
    R_\mathcal{F}(\M)=\max_\mathcal{E}\frac{p_\mathrm{succ}(\mathcal{E},\M)}{\max_{\N\in\mathcal{F}} p_\mathrm{succ}(\mathcal{E},\N)}\ .
\end{equation}
From Proposition \ref{prop:state discrimination} we immediately have $R_{\SP(d)} \leq (1/c)$. Similarly, from Result \ref{res:singleANC} we get $R_{\SP(d,2)}\leq 1/q$, where $q=1/8$ and $\SP(d,2)$ is the convex hull of POVMs on $\C^d$ that can be implemented with only  one auxiliary qubit (a generalization of this statement to measurements requiring $k$-dimensional ancillas follows from Theorem \ref{th:ancillas}).

We note that similar bounds between relative power of POVMs and measurements in $\SP(\C^d)$ or implementable by limited number of outcomes (measurements $\N^{(\beta)}$ from Result \ref{res:singleANC} have at most $d+1$ outcomes) hold also for other linear games like steering on Bell inequalities. Specifically, while in these contexts it is known that the number of outcomes of POVMs can reveal non-classical qualities of quantum states \cite{Kleinmann2016,Nguyen2020}, our simulation results imply that values of the corresponding linear functionals will be comparable.

\subsection{Limitations of POVMs in shadow tomography and related protocols}\label{sub:shadowTOM}
Another important application of our results concerns potential usefulness of POVMs in shadow tomography. Shadow tomography, introduced in \cite{HuangShadows2020}, is an important algorithmic primitive whose purpose is to estimate expectation values of multiple non-commuting observables on an (unknown) quantum state (see \cite{ReviewShadows23} for a recent review of the method and exposition in the broader context of quantum technologies). The basic premise of the technique is that by avoiding direct tomography of a quantum state, it is possible to simultaneously estimate many (even exponentially many) observables on an $n$-qubit state by performing only $\mathrm{poly}(n)$  measurement rounds. The most popular classical shadow protocols \cite{HuangShadows2020,Zhao21,Wan22} involve randomized implementation of projective measurements realized by appending randomly-chosen unitary transformation (chosen from the ensemble relevant to the particular application). However, in full generality, generalized measurements \cite{povmSHADOWfirst,povmSHADOWGhune,algoSHADOWnew} can be used to define classical shadows and in some cases offer an advantage over the standard protocols (see also \cite{MMO23,MMO24,MCO2024} for the complementary perspective in which incompatibility theory was proposed to estimate  multiple non-commuting quantities for qubit and fermionic systems). The following result shows that, for a broad class of classical shadow protocols, POVMs do not offer asymptotically unbounded (with the dimension of the Hilbert space) advantage compared to projectively simulable measurements.

\begin{prop}[Limitations of single-shot classical shadows based on generalized measurements]\label{prop:classSHADOW}
Let $\mathcal{O}=\lbrace O_i\rbrace_{i=1}^L$ be a collection of observables on $\C^d$ satisfying $\tr(O_i)=0$ for $i\in[L]$. Let $\M=(M_1,\ldots,M_n)$ be a POVM that can be used to estimate expectation values of observables $O\in\mathcal{O}$. Let $\hat{e}_O$ be an unbiased estimator of the expectation value of $O$, i.e. a real-valued function $\hat{e}_O:[n]\rightarrow \mathbb{R}$ satisfying
\begin{equation}
    \mathbb{E}\hat{e}_O = \sum_{i=1}^n \hat{e}_O(i) \tr(\rho M_i) = \tr(\rho O)\ ,
\end{equation}
for every state $\rho$. Let $\Delta_\M(O,\rho)=\mathbb{E}\hat{e}_{O}^2$ be the upper bound on the variance of $\hat{e}_O$. Then, for $c=0.02$  from Result \ref{res:projSIM}, a projectively simulable POVM $\N=\mathrm{\Phi}_c (\M)\in \SP(d)$ 
can be used to estimate expectation values of observables $O\in\mathcal{O}$ via estimators $\hat{e}'_{O}(k)\coloneqq\frac{1}{c} \hat{e}_O(k) $. Furthermore we have $ \max_\rho \Delta_\N(O,\rho)\leq (1/c^2) \max_\rho\Delta_\M(O,\rho)$.
\end{prop}

The relevance of the above proposition lies in the fact that $\Delta_\M(O,\rho)$ is typically used for assessing performance of the procedures based on classical shadows \cite{HuangShadows2020,ReviewShadows23}. This is because, by virtue of the median-of-means technique, $R\approx \max_{O \in \mathcal{O}} \Delta_\M (O,\rho) \log(L)/\epsilon^2$ is the sample complexity, i.e. number of copies of the state $\rho$ which are needed to estimate all expectation values of operators in $\mathcal{O}$ up to additive precision $\epsilon$. Our result states that for very general scenarios (involving arbitrary sets $\mathcal{O}$ of traceless observables), the worst-case sample complexity based bounds on estimation strategies using general POVMs can outperform the one based on randomization over projective measurements by only a constant factor.

It is possible to produce an analogous statement about the lack of advantage for general POVMs over measurements $\SP(d,2)$ that can be realized with a single auxiliary qubit (the possible sampling overhead $(1/c^2)$ is replaced by $1/q$, where $q=1/8$ is a constant from Result \ref{res:singleANC}).  We furthermore note that Proposition \ref{prop:classSHADOW} can be easily extended to estimation of nonlinear state functions that can be expressed as $f_X(\rho)= \tr(\rho^{\ot k} X)$, for some traceless observable $X$ acting on $(\C^d)^{\ot k}$.


\subsection{Power of POVMs implementable by a single auxiliary qubit} \label{sub:singleQubit}

The notion of simulation with postselection covered by Result \ref{res:singleANC} is particularly strong. Operationally, for a target POVM $\M$ and a quantum state $\rho$, the simulating POVM $\sum_\beta p_\beta \N^{(\beta)}$ generates, in a single experimental realization, a sample from output $i$ distributed according to the probability distribution $\lbrace p(i|\rho,\M)\rbrace$  with probability $q=1/8$ and with probability $7/8$ flags that the simulation protocol was not successful. There is a variety of tasks for which losing (on average) a constant fraction of samples is acceptable -- for such problems our POVM simulation strategy, utilizing just a single auxiliary qubit, achieves performance that matches the one of the optimal POVM $\M$ up to dimension-independent overhead. This time-ancilla space trade-off might be beneficial, as implementation of general $\M$ requires in general $N$ auxiliary qubits for an $N$-qubit system. From a variety of different possible scenarios we reference here (already mentioned) quantum state discrimination \cite{QuantumStateDisriminationRev}, shadow tomography \cite{ResTheoryMeas} and furthermore quantum state tomography \cite{Derka1998,Renes2004,OptTomography}, multi-parameter metrology \cite{Ragy2016,Szcyykulska2016,rafalMETRO} (in both frequentist and Bayesian approach).  We however note that for some specific applications, like port-based teleportation \cite{Ishizaka2009}, throwing away a constant fraction of samples undermines the performance of the measurement. To remedy this we need to increase the success probability $q$ of our  simulation protocol from Result \ref{res:singleANC} (so that it asymptotically equals $1$), which can be achieved by introducing additional ancillas, as explained in Theorem \ref{th:ancillas} in Section \ref{sec:Postselection}.

\subsection{Space-efficient circuit knitting from POVM simulation}\label{sec:circKNITTING}

The scheme from Result \ref{res:singleANC} can be interpreted as a (probabilistic) circuit knitting procedure from  Fig. \ref{fig:post}, i.e. as a method of simulating large quantum circuits by implementing smaller ones and applying appropriate classical post-processing (see e.g. \cite{knitting0,knitting1,knitting2} for the exemplary proposals of such schemes). Assume we apply our simulation procedure to a system of $N$ qubits so that $d=2^N$. Furthermore, assume that the target POVM $\M$ has rank-one effects with $d^2=4^N$ outcomes. Then, by virtue of Naimark theorem (Theorem \ref{th:Naimark}), it can be realised as a projective measurement, $\P^\M$, acting on $2N$ qubits (a system of dimension $d^2=4^N$) via 
\begin{equation}
    \tr(\rho M_i)=\tr\left(\rho \ot \ketbra{0_N}{0_N} P^{\M}_i\right)\ ,\ i\in[4^N]\ ,
\end{equation}
where $\rho$ an is arbitrary $N$-qubit state and $\ketbra{0_N}{0_N}$ is a fixed referential $N$-qubit state. Let $U$ be a unitary such that $U^\dag \ketbra{i}{i} U =P^{\M}_i$ , where $\{\ket{i}\}_{i=1}^{4^N}$ is a computational basis on $(\C^2)^{\ot N} \ot (\C^2)^{\ot N}$. From the statement of Result \ref{res:singleANC} we know that measurements $\N^{(\beta)}$ require only a single auxiliary qubit and a stochastic post-processing $\Q^{(\beta)}$ to be implemented (see the proof of Result \ref{res:singleANC} presented in Section \ref{sec:overview} for details). Stochastic operations $\Q^{(\beta)}$  transform probability distributions with $2\cdot 2^N$ outcomes into probability distributions on $4^N+1$ outcomes such that for for all $i\in[4^N]$  
\begin{equation}
    \tr(\rho N^{(\beta)}_i)= \sum_{l=1}^{2\cdot 2^N} q^{(\beta)}_{i|l}  \tr\left( \rho \ot \ketbra{0}{0} P^{(\beta)}_l\right)\ ,
\end{equation}
and
\begin{equation}
    \tr(\rho N^{(\beta)}_\emptyset)= \sum_{l=1}^{2\cdot 2^N} q^{(\beta)}_{\emptyset|l}  \tr\left( \rho \ot \ketbra{0}{0} P^{(\beta)}_l\right)\ ,
\end{equation}
where $\P^{(\beta)}$ is a projective measurement on $(\C^2)^{\ot N} \ot \C^2$ and $\ket{0}$ is a referential state of the auxiliary qubit. Let $U_\beta$ denote a unitary such that $U^{\dag}_\alpha \ketbra{l}{l} U_\beta = P^{(\beta)}_l$ for $l\in[2\cdot 2^N]$. Combining all these ingredients and Result \ref{res:singleANC} and using $\mathbf{x}=x_1 x_2 \ldots x_{2N}\in\{0,1\}^{2N} ,\mathbf{y}=y_1 y_2 \ldots y_{N+1} \in\{0,1\}^{N+1}$ to denote outcomes of multiqubit measurements, we obtain:

\begin{prop}[Unitary compression via POVM simulation] 
Let $U$ be a $2N$-qubit unitary circuit. Let $\rho$ be a $N$-qubit state For $q=1/8$ there exists a probability distribution $\{p_\beta\}$, stochastic transformations $\Q^{(\beta)}$ (transforming  classical states on $N+1$ bits into states on $2N$ bits)  and $N+1$-qubit unitaries $U_\beta$ such that
\begin{eqnarray}
    q \tr\left( U\rho \ot \ketbra{0_N}{0_N} U^\dagger \ketbra{\mathbf{x}}{\mathbf{x}}\right) \nonumber \\  =  \sum_{\beta} p_\beta \sum_{\mathbf{y}\in\{0,1\}^{N+1}} q^{(\beta)}_{\mathbf{x}|\mathbf{y}} \tr\left( U_\beta\rho \ot \ketbra{0}{0} U_\beta^\dagger \ketbra{\mathbf{y}}{\mathbf{y}} \right) \ ,
\end{eqnarray}
for $\mathbf{x}\in\{0,1\}^{2N}$. Additionally we have
\begin{equation}
   \sum_{\beta} p_\beta \sum_{\mathbf{y}\in\{0,1\}^{N+1}} q^{(\beta)}_{\emptyset|\mathbf{y}} \tr\left( U_\beta\rho \ot \ketbra{0}{0} U_\beta^\dagger \ketbra{\mathbf{y}}{\mathbf{y}} \right)=   1-q\ .
\end{equation}
\end{prop}
In other words, sampling from the output of a circuit $U$ on $\rho\ot\ketbra{0_N}{0_N}$ can be realized by: (i) sampling $\beta$ according to the probability distribution $\{p_\beta\}$, (ii) implementing a unitary $U_\beta$ on an $N+1$-qubit state $\rho \ot \ketbra{0}{0}$, (iii) performing a measurement in the computational basis on $N+1$ qubits and (iv) applying post-processing $\Q^{(\beta)}$ to the resulting outcome $\mathbf{y}$. Importantly, this method is guaranteed to generate a sample from the correct probability distribution with success probability $q=1/8$. See Fig. \ref{fig:post} for a graphical presentation of all the steps of the protocol. 

\begin{figure*}
\centering
\begin{adjustbox}{width=0.95\textwidth}
\includegraphics{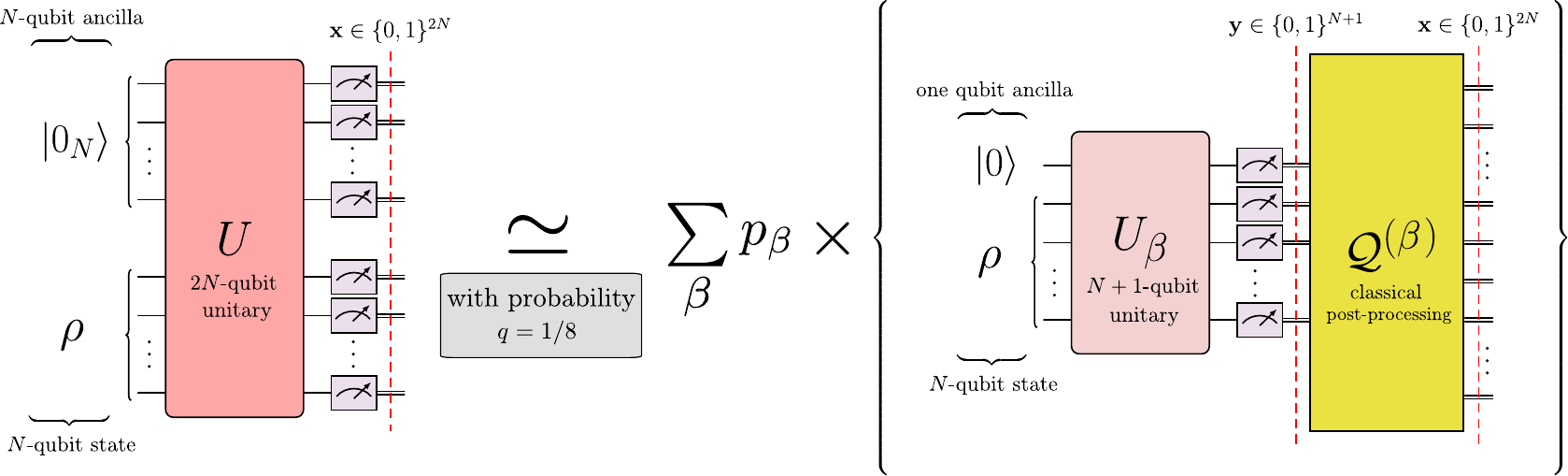}
\end{adjustbox}

\caption{A circuit knitting method originating from the POVM simulation protocol in Result \ref{res:singleANC}. The method realizes sampling from the output of a $2N$-qubit unitary $U$ on $\rho\ot\ketbra{0_N}{0_N}$, where $\rho$ is an $N$-qubit state and $\ket{0_N}$ is a pure state of an $N$-qubit ancilla. It proceeds by by: (i) sampling $\beta$ according to the probability distribution $\{p_\beta\}$, (ii) implementing an $N+1$-qubit unitary $U_\beta$ on a state $\rho \ot \ketbra{0}{0}$ (with $\ket{0}$ being a pure state of a single qubit), (iii) performing a measurement in the computational basis on $N+1$ qubits and (iv) applying post-processing $\Q^{(\beta)}$ to the resulting outcome $\mathbf{y}$ to return $\mathbf{x}$ or $\emptyset$ (a flag indicating that the protocol was unsuccessful). Importantly, the method generates a sample $\mathbf{x}$ from the correct probability distribution with success probability $q=1/8$.\label{fig:post}}
\end{figure*}

Our method is flexible and can be easily adjusted to handle the cases when the target quantum circuits $U$ act on an arbitrary number of qubits greater than $N$.  We want to emphasize that the above proposition does not provide an efficient method for constructing the probability distribution $\{p_\beta\}$, stochastic transformations $\Q^{(\beta)}$ and unitaries $U_\beta$. While some of these objects can be constructed efficiently for some classes of input unitaries $U$, we do not expect that in general it would be possible to devise an efficient algorithm for realization of our scheme for general poly($N$)-sized circuits on $N$ qubits. Nevertheless, we expect that for a moderate number of qubits and not too complicated circuits, our technique for circuit knitting can prove useful. We make a step towards this direction in Appendix \ref{app:randomPART}, where we show that a random choice of a partition of the set of outcomes (which enters into the definition of the simulation protocol, see Section \ref{sec:overview}) of an $N$-qubit extremal rank-one POVM gives success probability $q\approx 1/ N$. This guarantees that the sample complexity of the protocol (the number of trials needed to generate a sample from the correct distribution) still scales efficiently with system size. Using the partition that is guaranteed to exist via the solution of Kadison-Singer problem \cite{marcus-spielman-srivastava} (c.f. Section \ref{sec:Postselection}) gives $q=1/8$ but the complexity of finding such a partition can be exponential in $d$ (and hence doubly exponential in $N$).

\subsection{Hidden variable models for Werner and isotropic states}

Our next application concerns local hidden variable models for correlations originating from noisy pure states of two qudits, i.e states of the form 
\begin{equation}
\rho_\psi (t)  = t \ket{\psi}\bra{\psi}+  (1-t) \frac{1}{d^2}\I_d \ot \I_d \ ,
\end{equation}
where $\ket{\psi}$ is a fixed pure state on $\C^d\ot\C^d$ and $t\in[0,1]$ is typically called visibility. If $\ket{\psi}=\ket{\phi_d}=1/\sqrt{d} \sum_{1=1}^d \ket{i}\ket{i}$ is a maximally entangled state, then the corresponding noisy state is called an isotropic state and denoted by $\rho_{iso}(t)$.

A bipartite quantum state $\rho$ on $\H_A \ot \H_B$ is called local \cite{localREVIEW2014} (POVM-local) if for all quantum measurements $\M\in\POVM(\H_A)$, $\N\in\POVM(\H_B)$ and all outcomes $a,b$ we have
\begin{equation}\label{eq:locality}
    \tr(\rho M_a \otimes N_b) = \int_\Lambda d\lambda \, p(\lambda) \xi_A(a|\M,\lambda) \xi_B(b|\N,\lambda)\ ,
\end{equation}
where $\lambda$ denotes a hidden variable, $p(\lambda)$ is its distribution in the hidden variable space $\Lambda$, and $\xi_A,\xi_B$ are local response functions -- for fixed $\M$ and $\lambda\in\Lambda$, the collection $\lbrace{\xi_A (a|\M,\lambda) \rbrace}$ forms the probability distribution of the random variable $a$ (and analogously for the response function $\xi_B$). Physically, condition \eqref{eq:locality} means that for every possible experimental setting (given by the choice of local measurements) in a Bell scenario, outcome statistics of measurements performed on $\rho$ can be reproduced by a hidden variable model (specified by the distribution $p(\lambda)$ and response functions $\xi_A(a|\M,\lambda), \xi_B(b|\N,\lambda)$). A state is PM-local if condition \eqref{eq:locality} holds when $\M,\N$ are restricted to projective measurements. The interest in local models comes from the fact that there exist states which are entangled but cannot be used to violate Bell inequalities \cite{Barret2002}. 

The first local models for projective measurements for general dimension $d$ date back to the original paper of Werner \cite{Werner1989} who showed that so-called Werner states\footnote{Werner states are states of the form $\rho_W(t)=t \frac{1}{d(d-1}(\I_d\ot\I_d-V)+(1-t) \frac{1}{d^2} \I_d\ot\I_d$, where $t\in[0,1]$ and $V$ is the swap operator that permutes two factors of the tensor product $\C^d\ot \C^d$.} $\rho_W(t)$ are PM-local for $t\leq t^{PM}_W=\frac{d-1}{d}$. This approach to construction of local models has been further generalized  in \cite{Almeida2007,Wiseman2007} where PM-locality of isotropic states $\rho_{iso}(t)$ was shown for
\begin{equation}\label{eq:robustnessPM}
    t\leq t^{PM}_{iso}=\frac{\sum_{k=2}^d 1/k}{d} \approx \frac{\log(d)}{d}\ .
\end{equation}
Furthermore, \cite{Almeida2007} also proved that all noisy pure states $\rho_\psi(t)$ are PM-local for  
\begin{equation}
    t\leq t^{PM}_{\psi}=\frac{t^{PM}_{iso}}{(1-t^{PM}_{iso})(d-1)+1} \approx \frac{\log(d)}{d^2}\ .
\end{equation}

\begin{prop}[New local hidden variable model for noisy pure qudit states]\label{app:ISO}
 Let $c=0.02$ be the constant appearing in Result \ref{res:projSIM}. Then
 \begin{itemize}
     \item[(i)] States $\rho_{iso}(t)$ are POVM-local for $t\leq c\ t^{PM}_{iso}\approx \frac{c \log(d)}{d}$.
     \item[(ii)] For any pure state $\ket{\psi}$ on $\C^d \ot \C^d$ states $\rho_{\psi}(t)$ are POVM-local for $t\leq  \frac{c\ t^{PM}_{iso}}{(1-c\ t^{PM}_{iso})(d-1)+1} \approx \frac{c \log(d)}{d^2}$.
 \end{itemize}
\end{prop}

\noindent
Prior to this work there existed a gap in the range of visibilities for which isotropic states admitted local models for projective and generalized measurements.  Previously known bounds for  general POVMs were much weaker \cite{Almeida2007}  and implied POVM-locality for $t^{\POVM}_{iso,old}=\Theta(1/d)$ and $t^{\POVM}_{\psi,old}=\Theta(1/d^2)$ for isotropic and general noisy states respectively.  

Our results offer asymptotic  improvements in the range of visibilities for which  noisy two qudit states are local. The proof of Proposition \ref{app:ISO} follows from Result \ref{res:projSIM} and a known technique of moving noise $\mathrm{\Phi}_c$ form POVMs to a state in order to construct new hidden variable models for POVMs \cite{Bowels2015,Oszmaniec17,Hirsch2017betterlocalhidden} (the details are presented in Appendix \ref{app:Applications}). The same technique can be applied to Werner states $\rho_W(t)$, showing that they are POVM-local for   $t \leq c\ t^{PM}_W = \Theta(1)$. This improves over the Barret model from \cite{Barret2002} that proved POVM locality for $t\leq t^{POVM}_{W,old}\approx\frac{3}{e d}$. However, in this case \cite{Nguyen2020}  already provided an improved model for POVMs, proving that $\rho_W (t)$ is local for $t\leq 1/e$. We note however that our construction of the model, after one accepts Result \ref{res:projSIM}, is qualitatively much simpler than the one given therein. 

\subsection{Incompatibility robustness of POVMs}

Measurement incompatibility, one of the defining  features of quantum theory, states that certain POVMs cannot be simultaneously measured. Incompatibility underpins many nonclassical aspects of quantum physics, such as uncertainty relations or nonlocality, and has applications in numerous quantum protocols and subroutines (see e.g \cite{RevInco} for a recent review). 

For projective measurements $\P,\P'$ incompatibility is equivalent to non-commuting of measurement operators $P_i, P'_j$. However,  a collection of general POVMs $\{\M^{(x)}\}_{x\in X}$ can be jointly measurable in the sense that there exists a parent POVM $\mathbf{G}=(G_\lambda)_{\lambda\in \Lambda}$ that can simulate each $\M^{(x)}$, i.e. for every $x\in X$ we have $\M^{(x)}=\mathcal{Q}^{(x)}(\mathbf{G})$, for some classical post-processing operations $\mathcal{Q}^{(x)}$. Generally speaking, every collection of POVMs becomes jointly measurable once sufficient amount of noise is added to the measurements in question, with the paradigmatic example being noisy Pauli measurements: $M^{X,\eta_x}_\pm=\frac{1}{2}(\I\pm \eta_x X)$, $M^{Y,\eta_Y}_\pm=\frac{1}{2}(\I\pm \eta_Y Y)$ and $M^{Z,\eta_Z}_\pm=\frac{1}{2}(\I\pm \eta_Z Z)$, that are jointly measurable if an only if  $\eta^2_X+\eta^2_Y+\eta^2_Y\leq 1$ \cite{Teiko08}. There has been a significant interest in recent years to quantify noise tolerance for incompatibility for pairs of measurements \cite{Farkas19}, mutually unbiased bases \cite{Designolle2019}, or collections of measurements exhibiting symmetries \cite{NguyenJM2020} (including Majorana-fermion observables \cite{MCO2024}). 

In \cite{Almeida2007,Wiseman2007} it was proven that for $t\leq t^{PM}_{iso}$ (where $t^{PM}_{iso}$ is given by Eq. \eqref{eq:robustnessPM}) all measurements of the form $\mathrm{\Phi}_t(\P)$, where $\P\in\PP(\C^d)$, are jointly measurable. Our Result \ref{res:projSIM}  allows to straightforwardly generalize this result to noisy POVMs.

\begin{prop}[Improved compatibility region for noisy POVMs]
    Let $c=0.02$ be the constant appearing in Result \ref{res:projSIM}. Furthermore assume
\begin{equation}\label{eq:compbound}
    t\leq c\ t^{PM}_{iso}  \approx \frac{c \log(d)}{d}\ .
\end{equation}    
Then all POVMs of the form $\mathrm{\Phi}_t(\M)$ are jointly-measurable, where $\M$ ranges over all POVMs in $\C^d$.
\end{prop}
\begin{proof}
Clearly, noisy versions $\mathrm{\Phi}_{\tau}(\N)$ of all projectively-simulable POVMs $\N\in\SP(\C^d)$ are jointly measurable for $\tau\leq t^{PM}_{iso}$. Since $\mathrm{\Phi}_c(\M)\in \SP(\C^d)$ for all $\M\in\POVM(\C^d)$ we get that all measurements of the form $(\mathrm{\Phi}_\tau\circ\mathrm{\Phi}_c)(\M)=\mathrm{\Phi}_{c\tau}(\M)$ are jointly measurable for $\tau\leq t^{PM}_{iso}$.
\end{proof}

The inequality Eq.\eqref{eq:compbound} is asymptotically tight (in terms of dependence of the critical visibility on $d)$ -- this is because \cite{Wiseman2007} proved that noisy projective measurements $\mathrm{\Phi}_t(\P)$ are incompatible for $t>t^{PM}_{iso}$.  The occurrence of the same bound $t^{PM}_{iso}$ for the critical visibility for this problem and local models for projective measurements for isotropic states is due to a one-to-one relation between incompatibility of noisy  measurements on $\C^d$ and so-called steering of isotropic states \footnote{Steering is a form of quantum correlation in bipartite systems weaker than non-locality. It assumes a partial trust in measurements performed in one subsystem participating in the Bell scenario.}. This is an example of a more general relation between steering and incompatibility established in \cite{SteeringUola2015}. Recently it was proven  \cite{stearing24one,stearing24two} that for two-qubit isotropic states POVMs are as powerful as projective measurements in revealing steering and that the critical visibility for this characteristic of a state equals $t_{iso,steer}=1/2$ (which also proves that critical visibility for joint measurability of all qubit measurements equals $1/2$). This suggests that perhaps the constant $c$ in the inequality \eqref{eq:compbound} can be dropped and that incompatibility robustness of noisy POVMs matches that of projective measurements.

\subsection{Discussion and open problems}\label{sec:discussion}

In this work we gave a surprising structural result on the power of generalized measurements in quantum information theory. Namely, we have shown that measurements implementable with no ancillas or low dimensional-ancillas can offer similar performance to general POVMs in a variety of applications. We realized this by concrete simulation protocols that build on previous results on POVM simulability \cite{SMO2022}, newly-introduced dimension-deficient Naimark extension theorem (cf. Theorem \ref{th:dimentionDefitient}) and the use of the solution to the celebrated Kadison-Singer conjecture (cf. Theorem \ref{th:solutionKS}).

There is a number of interesting open problems that originate from our work. First, we expect that it is possible to extend our simulation techniques to other relevant objects: quantum channels, instruments and combs. Second, it would be interesting to explore whether the solution to the Kadison-Singer problem (Theorem \ref{th:solutionKS} and its generalizations, see \cite{bownik-casazza-marcus-speegle, ravichandran-leake, branden, bownik, xu-xu-zhu}) has other applications in quantum information science.  Third, our protocol is not constructive in the sense of not providing a circuit description of unitaries realizing sub-POVMs $\N^{(\beta)}$ that simulate a target measurement $\M$ with (constant) postselection probability. It is therefore natural to investigate the possibility of turning our method into an algorithm, at least for some restricted classes of quantum measurements (or circuits realizing them). As mentioned in Section \ref{sec:circKNITTING}, a positive result in that direction can be potentially useful as a new circuit knitting strategy. Another interesting problem would be to identify, in every dimension $d$, POVMs that are hardest to simulate by restricted classes of  measurements, and to compute the optimal value of  constants $c$ and $q$ for which the simulation is possible in any finite $d$ (we expect that actual optimal values of these constants are much higher than the ones given in Results \ref{res:singleANC} and \ref{res:projSIM}). Finally, it would be  desirable to understand the computational complexity of deciding whether a given POVM can be approximated by a convex combination of projective measurements, following earlier works concerning separable states \cite{Gharibian2010} and mixed-unitary channels \cite{Watrous2020}.

\tikzstyle{block} = [rectangle, draw, fill=blue!20, rounded corners, minimum width=10em, minimum height=2em, inner sep=4 pt, outer sep=3 pt]

\tikzstyle{resultblock} = [rectangle, draw, fill=green!30, rounded corners, minimum width=10em, minimum height=2em, inner sep=4 pt, outer sep=3 pt]

\tikzstyle{dashedblock} = [rectangle, draw, dashed,fill=red!20, 
     text centered, rounded corners, minimum width=10em, minimum height=2em, inner sep=4 pt, outer sep=3 pt]
\tikzstyle{line} = [draw, -{Latex[length=2mm]}]

\tikzstyle{dashline} = [draw, -{Latex[length=2mm]}]

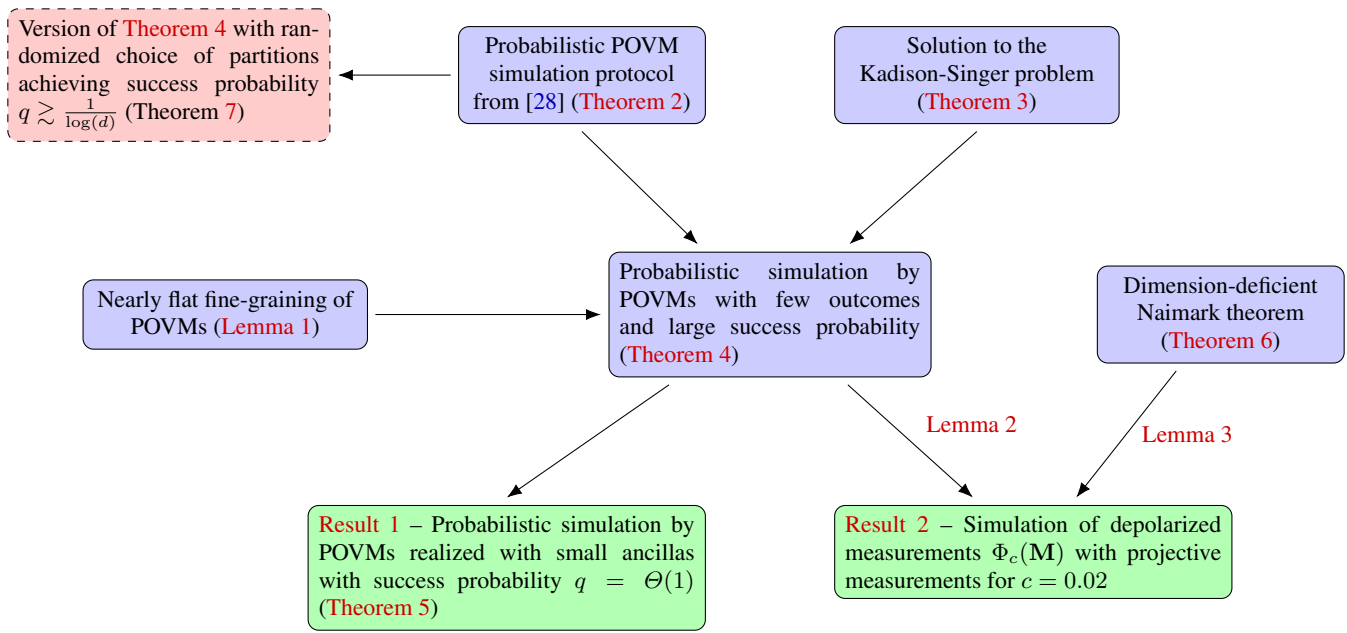
\begin{figure*}
    \centering
    \begin{tikzpicture}[node distance = 1cm, auto]
    
    \node [block] (Theorem4) {\begin{minipage}{4cm}
        \noindent   \justifying Probabilistic simulation by POVMs with few outcomes and large success probability (\cref{th:PostSimulation})
    \end{minipage}};

    \node [block, left=3cm of Theorem4] (nearlyFlat) {\begin{minipage}{3.5cm}
        Nearly flat fine-graining of POVMs (\cref{lem:flatPOVM})
    \end{minipage}};
    
    \node [block, above left=1.5cm and 3cm of Theorem4.north east] (simProtocol22) {\begin{minipage}{3cm}
        Probabilistic POVM simulation protocol from \cite{SMO2022} (\cref{th:OldProtocol})
    \end{minipage}};

    \node [block, above right=1.5cm and 3cm of Theorem4.north west] (KSProblem) {\begin{minipage}{3.5cm}
        Solution to the Kadison-Singer problem (\cref{th:solutionKS})
    \end{minipage}};
    
    \node [resultblock, below left=1.5cm and 3cm of Theorem4.south east] (simSmallAncillas) {\begin{minipage}{5cm}
        \noindent   \justifying \cref{res:singleANC} -- Probabilistic simulation by POVMs realized with small ancillas with success probability $q = \Theta(1)$ (\cref{th:ancillas})
    \end{minipage}};
    
    \node [resultblock, below right=1.5cm and 3cm of Theorem4.south west] (simPhi) {\begin{minipage}{5cm}
        \noindent   \justifying \cref{res:projSIM} -- Simulation of depolarized measurements $\mathrm{\Phi}_c(\M)$ with projective measurements for $c=0.02$
    \end{minipage}};
    
    \node [block, right=2cm of Theorem4] (deficientNaimark) {\begin{minipage}{3cm}
        Dimension-deficient Naimark theorem (\cref{th:dimentionDefitient})
    \end{minipage}};
    
    \node [dashedblock, left = 1.5cm of simProtocol22] (chernoff) {\begin{minipage}{4cm}
        \noindent   \justifying Version of \cref{th:PostSimulation} with randomized choice of partitions achieving success probability $q \gtrsim \frac{1}{\log (d)}$ (Theorem \ref{th:PostSimulation-randomized})
    \end{minipage}};
      
    
    \path [line] (nearlyFlat) -- (Theorem4);
    
    \path [line] (simProtocol22.south) -- (Theorem4);
    
    \path [line] (KSProblem.south) -- (Theorem4);
    
    
    \path [line] (Theorem4) -- (simSmallAncillas.north);
    
    \path [line] (Theorem4) -- (simPhi) node[midway, above right=0.01cm and 0.1cm] {\cref{lem:nearlyPROJpost}};

    \path [line] (deficientNaimark) -- (simPhi) node[midway,right=0.1cm] {\cref{lem:simnearlyPROJ}};

    \path [dashline] (simProtocol22) -- (chernoff);


\end{tikzpicture}
    
    \caption{Structure of the main results of the paper.  Some of the arrows are labelled by auxiliary lemmas used in the proof of the respective result. We start with a POVM to be simulated $\M \in \POVM(\C^d)$. First, classical post-processing is used to replace $\M$ with a \emph{fine-grained} POVM $\M'$ whose outcomes have nearly equal magnitude (Lemma \ref{lem:flatPOVM}). Our central intermediate result, Theorem \ref{th:PostSimulation}, states that such POVMs can be simulated with high success probability using a convex combination of POVMs with only $\Theta(d)$ outcomes. The simulation relies on an explicit probabilistic simulation protocol from \cite{SMO2022} (Theorem \ref{th:OldProtocol}) and the solution to the Kadison-Singer problem (Theorem \ref{th:solutionKS}). From this we derive our first main result, Result \ref{res:singleANC} -- the POVM $\M$ can be simulated with high probability by projective measurements using only a low-dimensional ancilla. Our second result, Result \ref{res:projSIM}, states that $\mathrm{\Phi}_c(\M)$, a noisy version of $\M$ with constant noise parameter, can be simulated by projective measurements that do not require any ancilla. To prove this, we first show that the fine-grained POVM $\M'$ is simulable by \emph{nearly projective} measurements, as stated in Lemma \ref{lem:nearlyPROJpost}. Noisy versions of such measurements are then shown to be easily simulable by projective measurements (Lemma \ref{lem:simnearlyPROJ}), which requires an additional result, the dimension-deficient Naimark theorem (Theorem \ref{th:dimentionDefitient}). An additional result (Theorem \ref{th:PostSimulation-randomized}) shows that one can obtain simulation success probability $\Theta(1/\log (d))$ with the use of random partitions (which can be efficiently generated).}
    \label{fig:outline}
\end{figure*}

\section{Overview of the proof of the main results}\label{sec:overview}

In this part we present the outline of the proof of our main findings  -- Results \ref{res:singleANC} and \ref{res:projSIM}. The proof relies on explicit simulation protocol and consists of several steps and simplifications -- see Fig. \ref{fig:outline} for a schematic presentation of the argument.  Along the way we introduce additional Lemmas whose proofs are given in latter parts of the article. The starting point is a target POVM $\M\in\POVM(\C^d)$ about which we do not assume anything except having a finite number of outcomes (our results easily extend to general POVMs with continuous number of outcomes, cf. Remark \ref{rem:OutcomeNumber}). 




\emph{Step 1 -- Special form of a POVM.} We first use post-processing to reduce the structure of the target POVM in such a manner that  the POVM to-be simulated has rank-one effects of \emph{nearly equal magnitude}. It is well known that every POVM can be obtained by coarse-graining of a POVM with rank $1$ effects. However, the fact that magnitude of the effects of fine grained POVM can be made (nearly) uniform is to our knowledge new.

\begin{lem}[Nearly flat fine-graining of POVMs]\label{lem:flatPOVM} \  For every $\delta\in(0,1)$ and for every $\M\in\POVM(\C^d)$ there exist $\epsilon_\ast >0$ such that for all $\epsilon\in(0,\epsilon_\ast)$ there  exists a POVM $\M'$ and a stochastic map $\Q$ such that $\M=\Q(\M')$ and $M'_i=\alpha_i \ket{\psi_i}\bra{\psi_i}$, $\frac{\max_i \alpha_i}{\min_i \alpha_i} \leq 1 + \delta$, and furthermore $\max_i \alpha_i \leq \epsilon$.  
\end{lem}
\noindent The proof of the Lemma is given in Appendix \ref{app:aux}. In what follows the flatness parameter $\delta>0$ will play a relatively minor role and the right intuition is to think that we can set $\delta=0$. However, in order to keep our analysis elementary (that is, without invoking functional theoretic details) we have decided to keep nonzero $\delta$, and to phrase all intermediate auxiliary results with this parameter being present. At the end of the argument we take the limit $\delta \to 0$.

\emph{Step 2 -- Simulation with postselection via simpler POVMs.}\ \ We show that a POVM with ``nearly flat'' effects of rank 1 can be simulated with postselection by \emph{nearly projective} quantum measurements.

\begin{lem}[Simulation with postselection via nearly projective measurements]\label{lem:nearlyPROJpost} Let $\M'=\left(\alpha_1 \ket{\psi_1}\bra{\psi_1},\ldots,\alpha_n \ket{\psi_n}\bra{\psi_n}\right)\in \POVM(\C^d)$ be such that $\frac{\max_i \alpha_i}{\min_i \alpha_i} \leq 1 + \delta$, for $\delta\in(0,1]$. Then there exist POVMs $\N^{(\beta)}$ with outcome space $[n]\cup\lbrace\emptyset\rbrace$ and a probability distribution $\lbrace p_\beta \rbrace$ such that 
\begin{enumerate}
    \item[(i)] The convex combination $\L= \sum_{\beta} p_\beta \N^{(\beta)}$ simulates $\M'$ with postselection
    \begin{equation}\label{eq:nearlyPROJ}
       \sum_\beta p_\beta \N^{(\beta)} = \left(q \M', (1-q) \I_d \right)\ ,
    \end{equation}
    for $q\geq\frac{0.068}{1+\delta}$.
    \item[(ii)] Each POVM $\N^{(\beta)}$ is nearly projective in the sense that 
     \begin{equation}
      N^{(\beta)}_i = 
       \begin{cases}
        A_i \ket{\psi_i}\bra{\psi_i} & \text{if}\ i\in S_\beta  \\
        \I -  \sum_{j\in S_\beta} A_j \ket{\psi_j}\bra{\psi_j}& \text{if}\   i=\emptyset \\
        0 & \text{otherwise} \\ 
    \end{cases}, 
    \end{equation}
    where $A_i\geq \frac{0.47}{1+\delta}$  and $S_\beta \subset [n]$ satisfies $|S_\beta| \leq d/2$. 

\end{enumerate}
\end{lem}\noindent

The condition $|S_\beta|\leq d/2$ is necessary so that later on we will be able to employ Lemma \ref{lem:simnearlyPROJ}, stating that noisy versions of measurements $\N^{(\beta)}$ appearing in Lemma \ref{lem:nearlyPROJpost} are simulable by projective measurements. The proof of Lemma \ref{lem:nearlyPROJpost} is given in Section \ref{sec:Postselection} and follows from Theorem \ref{th:PostSimulation} given therein. 

Theorem \ref{th:PostSimulation} is our central result and utilizes the solution to the Kadison-Singer problem from \cite{marcus-spielman-srivastava}, which guarantees that the probabilistic POVM protocol developed in \cite{SMO2022} is capable (for a suitable choice of the partition $\mathcal{S}$ of the set $[n]$ of outcomes of a POVM) of simulating  measurements with nearly flat effects with a convex combination of measurements with number of outcomes $|S_\beta|+1=\Theta(d)$, while maintaining success probability $q=\Theta(1)$. Crucially, the measurements appearing in the simulation protocol have mostly effects of rank one, which allows them to be realized by a single auxiliary qubit and projective measurements. This directly underpins Result \ref{res:singleANC}, whose extended version is given and proven in Theorem \ref{th:ancillas} in Section \ref{sec:Postselection}.

\emph{Step 3 -- PM-simulability of noisy nearly projective measurements.}
It turns out that noisy versions of measurements $\N^{(\beta)}$ that appeared in \emph{Step 2} are quite easy to simulate by convex combinations of projective measurements.
\begin{lem}[Noisy nearly projective measurements are easy to simulate by projective measurements]\label{lem:simnearlyPROJ}
    Let $l\leq d/2$ and let $\N=(N_1,\ldots,N_{l+1})\in\POVM(\C^d)$ be a POVM of the form 
    \begin{equation}
         N_i = 
       \begin{cases}
        A_i \ket{\psi_i}\bra{\psi_i} & \text{if}\ i\in [l]  \\
        \I_d -  \sum_{j=1}^l A_j \ket{\psi_j}\bra{\psi_j}& \text{if}\   i=l+1 \\
        \end{cases}\ .
    \end{equation}
    Then $\mathrm{\Phi}_{t}(\N)\in\SP(d)$ for
    \begin{equation}
    t= \min_{i\in[l]} \frac{|W^\perp| A_i}{|W|(1-A_i)+|W^\perp|}\ ,
\end{equation}
where $W\coloneqq\mathrm{span}_{\C}\{\ket{\psi_i}\ |\ i\in[l]\}$ and $W^\perp$ is its orthogonal complement.
\end{lem}

\noindent
This result relies on dimension-deficient Naimark dilation theorem (Theorem \ref{th:dimentionDefitient} in Section \ref{sec:defitientNAIMARK}) and simple but tedious algebraic manipulations. For this reason the complete proof is given in Appendix \ref{app:dilationTOwhitenoise}. By applying the result to POVMs $\N^{(\beta)}$ from Lemma \ref{lem:nearlyPROJpost} we obtain that their noisy versions are projectively simulable for $t_{NP}\geq\frac{0.3}{1+\delta}$.

\emph{Step 4 -- Incorporate post-processing to show that noisy versions of $\M'$ and $\M$ are PM-simulable.}

Let $\L$ be a POVM appearing in the formulation of Lemma \ref{lem:nearlyPROJpost}. By applying $\mathrm{\Phi}_{t_{NP}}$ to both sides of Eq.\eqref{eq:nearlyPROJ} and utilizing Lemma \ref{lem:simnearlyPROJ}  we get
\begin{equation}
    \mathrm{\Phi}_{t_{NP}}(\L)=\sum_{\beta} p_{\beta} \mathrm{\Phi}_{t_{NP}}(\N^{(\beta)}) \in \SP(d)\ .
\end{equation}
It is now straightforward to note that for all $t\in[0,1]$ $\mathrm{\Phi}_t(\L)$ simulates $\mathrm{\Phi}_t (\M')$ with success probability $q$ i.e. $\mathrm{\Phi}_t(\L)=(q\mathrm{\Phi}_t(\M'),(1-q)\I_d)$. Therefore, by applying to $\mathrm{\Phi}_{t_{NP}}(\L)$ the post processing $\mathcal{Q}'=\{q'_{i|j}\}$ such that $q'_{i|j}=\delta_{ij}$ for $i,j\in[n]$ and $q'_{i|j}=\tr(M'_i)/d$, for $j=\emptyset$ and $i\in[n]$, we get that 
\begin{equation}
  \mathrm{\Phi}_q\circ \mathrm{\Phi}_{t_{NP}} (\M') = \mathcal{Q}'( \mathrm{\Phi}_{t_{NP}}(\L))=\sum_\beta p_\beta \mathcal{Q}' \left(\mathrm{\Phi}_{t_{NP}}\left(\N^{(\beta)}\right)\right)\ ,   
\end{equation}
from which it follows that for  $c=q\cdot t_{NP}$ the measurement $\mathrm{\Phi}_{c}(\M')$ is projectively simulable, since $\mathcal{Q}' \left(\mathrm{\Phi}_{t_{NP}}\left(\N^{(\beta)}\right)\right)\in\SP(d)$. Finally, we note that for every stochastic map $\mathcal{Q}$, any $t\in[0,1]$ and any POVM $\M$ we have $\mathcal{Q}(\mathrm{\Phi}_t(\M))=\mathrm{\Phi}_t\left(\mathcal{Q}(\M))\right)$. Consequently, using $\M=\mathcal{Q}(\M')$ (where $\mathcal{Q}$ is the stochastic map appearing in Lemma \ref{lem:flatPOVM} of \emph{Step 1}), we get 
\begin{equation}
\mathrm{\Phi}_c(\M)=\mathcal{Q}\left(\mathrm{\Phi}_c(\M)\right)=    \sum_\beta p_\beta \mathcal{Q} \circ \mathcal{Q}' \left(\mathrm{\Phi}_{t_{NP}}\left(\N^{(\beta)}\right)\right) ,
\end{equation}
which  concludes the proof of Result \ref{res:projSIM} for $c=q\cdot t_{NP}\geq \frac{0.0204}{(1+\delta)^2}$,  since POVMs $\mathcal{Q} \circ \mathcal{Q}' \left(\mathrm{\Phi}_{t_{NP}}\left(\N^{(\beta)}\right)\right)$ are projectively simulable and $\delta$ can be chosen small enough to ensure $c\geq 0.02$.

\begin{rem}
In the above reasoning there is room for flexibility regarding tradeoffs between different parameters. Specifically, one can decide to simulate the POVM $\M'$ via POVMs with fewer outcomes. This generally results in smaller success probability of simulation $q$ and smaller lower bound on $A_i$. At the same time smaller group size $|S_\beta|$ implies smaller dimension $|W|$ for which Lemma \ref{lem:simnearlyPROJ} has to be applied, which can increase $t_{NP}$. In fact, the specific choice of the group size $|S_\beta|$ (or more precisely, the implicit parameter $C$ which controls it, cf. proof of Lemma \ref{lem:nearlyPROJpost}) was made so as to maximize the product $c=q\cdot t_{NP}$.
\end{rem}

\section{Probabilistic simulation by measurements requiring ancillas with limited dimension}\label{sec:Postselection}

In what follows we give the proof of Lemma \ref{lem:nearlyPROJpost}, which states that rank-one POVMs with nearly flat effects can be simulated by \emph{nearly projective} measurements with constant (dimension independent) success probability. Note that nearly projective measurements can be realized with only a single auxiliary qubit. We also provide generalizations of this result for higher dimensions of the ancilla. We will make use of the POVM simulation technique introduced in \cite{SMO2022}, which gives a recipe to probabilistically simulate any POVM with measurements having smaller number of outcomes.

\begin{thm}[Simulation protocol from \cite{SMO2022}]\label{th:OldProtocol}
   Let $\M=(M_1,\ldots,M_n)$ be an $n$-outcome POVM on $\C^d$. Let $\mathcal{S} = \lbrace S_\beta\rbrace_{\beta=1}^r $ be a partition of $[n]$ into disjoint subsets. For a fixed $\beta$ we set $\lambda_\beta= \| \sum_{i\in{S_\beta}} M_i\|$ and define a POVM $\N^{(\beta)}$ by 
   \begin{equation}\label{eq:oldSIMPLeeffects}
       N^{(\beta)}_i = 
       \begin{cases}
        \lambda_\beta^{-1} M_i & \text{if}\ i\in S_\beta  \\
        \I - \lambda_\beta^{-1} \sum_{i\in S_\beta} M_i& \text{if}\   i=\emptyset \\
        0 & \text{otherwise} \\
    \end{cases} \ 
   \end{equation}
 Set $p_\beta = \lambda_{\beta}/(\sum_{\beta=1}^r\lambda_{\beta})$. Then the POVM $\L=\sum_{\beta=1}^r p_\beta \N^{(\beta)}$ simulates the POVM $\M$ with success probability
 \begin{equation}\label{eq:explicitQsucc}
     q(\M, \mathcal{S})= \left(\sum_{\beta=1}^r \left\| \sum_{i\in{S_\beta}} M_i\right\| \right)^{-1}\ ,
 \end{equation}
 i.e. $\L=(q(\M, \mathcal{S})\M,(1-q(\M, \mathcal{S}))\I_d)$.
\end{thm}
The potential advantage of the above protocol lies in the fact that for rank-one POVMs $\M$ the dimension of the  Hilbert space needed to implement the Naimark dilation of each of the measurements $\N^{(\beta)}$ is upper bounded by $|S_{\beta}|+d-1$.  However, it is a priori difficult to guarantee large success probability $q(\M, \mathcal{S})$  while maintaining small sizes of subsets $S_\beta$. Furthermore,  optimizing over partitions $\mathcal{S}$ with bounded subset size $|S_\beta|$ is a hard computational problem. The following Theorem \ref{th:PostSimulation} nevertheless ensures that for nearly flat rank-one POVMs there always exists a ``good partition''. 

The key result on which we build is the solution to the Kadison-Singer problem. Translated into the language of POVMs it reads as follows. 

\begin{thm}[{\cite[Corollary 1.5]{marcus-spielman-srivastava}}]\label{th:solutionKS}
    Let $\M = (M_1, \ldots, M_n)$ be an $n$-outcome  POVM on $\C^d$ , with each $M_i$ having rank-one effects and satisfying $\| M_i \| \leq \epsilon$, $i=1,\ldots,n$. Then for any $r \geq 1$ there exists a partition $\lbrace S_\beta\rbrace_{\beta=1}^r$ of $[n]$ into disjoint subsets such that for each $\beta = 1, \ldots, r$ we have
\begin{equation}\label{eq:KSbound}
    \left\| \sum_{i\in{S_\beta}} M_i \right\| \leq \frac{1}{r}\left( 1 + \sqrt{r\epsilon}  \right)^2\ .
\end{equation}

\end{thm}
\noindent Let us remark that in general the $O(\sqrt{\epsilon})$ dependence on $\epsilon$ in the upper bound cannot be improved (see \cite[Example 7]{weaver}).

By employing the above we obtain the following result on simulation of nearly flat POVMs with rank-one effects.  

\begin{thm}\label{th:PostSimulation}
    Let $\M = (M_1,\ldots,M_n)$ be an $n$-outcome POVM on $\C^d$, with $M_i = \alpha_i \ketbra{\psi_i}{\psi_i}$, $ \Tilde{\epsilon} \leq \alpha_i \leq \epsilon$. Fix $r \geq 1$ and let $C = r \epsilon$. Then there exists a partition $\mathcal{S}=\lbrace S_\beta\rbrace_{\beta=1}^r$ of $[n]$ into disjoint subsets such that
\begin{equation}\label{eq:KSqsucc}
    q(\M,\mathcal{S})\geq  \frac{1}{(1 + \sqrt{C})^2}\ ,
\end{equation}
where $q(\M,\mathcal{S})$ is given in Eq. \eqref{eq:explicitQsucc} and describes success probability of the simulation protocol from Theorem \ref{th:OldProtocol}. Furthermore, for all $\beta = 1,\ldots,r$
\begin{equation}\label{eq:upCARD}
    |S_{\beta}| \leq d \frac{\epsilon}{\Tilde{\epsilon}} (1 + 1/\sqrt{C})^2\ .
\end{equation}

\end{thm}

\begin{proof}
  Since $\M$ is a POVM with rank-one effects and $\| M_i \| \leq \epsilon$, we can apply Theorem \ref{th:solutionKS} to obtain for every $r\geq1$ that there exists a partition $\lbrace S_\beta\rbrace_{\beta=1}^r$ of $[n]$  into disjoint subsets such that for all $\beta=1,\ldots, r$  we have
  \begin{equation}\label{eq:normbBoundKS}
        \left\| \sum_{i\in{S_\beta}} M_i \right\| \leq \frac{1}{r}\left( 1 + \sqrt{r\epsilon}  \right)^2 ,
\end{equation}
        from which the Eq. \eqref{eq:KSqsucc} follows after summing over $\beta$ and utilizing  \eqref{eq:explicitQsucc}.
        
        In order to prove \eqref{eq:upCARD} we observe that \eqref{eq:normbBoundKS} is equivalent to
        \begin{equation}
            \sum\limits_{i \in S_{\beta}} M_i  \leq \frac{1}{r}\left( 1 + \sqrt{r\epsilon}  \right)^2\mathbb{I}_d.
        \end{equation}
        Upon taking trace of both sides we obtain
\begin{equation}
    \sum\limits_{i \in S_{\beta}} \alpha_i \leq \epsilon\left( 1 + \frac{1}{\sqrt{C}}  \right)^2 d\ .
\end{equation}
Using $\alpha_i \geq \Tilde{\epsilon}$, we finally obtain
       \begin{equation}
            \Tilde{\epsilon} |S_{\beta}| \leq \epsilon \left( 1 + \frac{1}{\sqrt{C}}  \right)^2 d\ ,
       \end{equation}
        which concludes the proof. 
\end{proof}

We remark that Theorem \ref{th:PostSimulation} does not provide an effective method for finding a partition $\mathcal{S}$ for which inequality \eqref{eq:KSqsucc} holds. This is due to the nonconstructive nature of the result of Marcus, Spielman and Srivastava. We leave open the problem of gauging the complexity of finding a good partitions $\mathcal{S}$ (i.e. partitions for which $q(\M,\mathcal{S})$ is large while maintaining $|S_\beta| =\Theta(d)$). In Appendix \ref{app:randomPART} we show (see Theorem \ref{th:PostSimulation-randomized}) that random (and thus efficient to find) partitions $\mathcal{S}$ yield success probability that decays like $1/\log(d)$ while maintaining $|S_\beta| =\Theta(d)$.

We will now use Theorem \ref{th:PostSimulation} to prove Lemma \ref{lem:nearlyPROJpost}.

\begin{proof}[Proof of Lemma \ref{lem:nearlyPROJpost}]

We first note that the statement of Theorem \ref{th:PostSimulation} is qualitatively very similar to that Lemma \ref{lem:nearlyPROJpost}. Additionally, the assumption $\max_i \alpha_i / \min_i \alpha_i \leq 1+\delta$ translates to $\epsilon/\tilde{\epsilon}\leq 1+ \delta$. Therefore, we only need to control magnitudes $A_i$ of effects $N^{(\beta)}_i$, for $i\in S_\beta$ and sizes of subsets $S_\beta$.

We start by bounding the size of subsets $S_\beta$. Unfortunately, Eq. \eqref{eq:upCARD} gives a bound on $|S_\beta|$ which is larger than $d$ for any $C$, $\epsilon>\tilde{\epsilon}>0$. 
To ensure small sizes of subsets we set $C$ to a fixed value\footnote{Formally, $C$ is not a free real parameter but is a multiple of $\epsilon$. However, by virtue of Lemma \ref{lem:flatPOVM}, $\epsilon$ can be chosen to be arbitrary small and hence $C$ can be effectively regarded as an unconstrained real parameter.} and consider a subpartition $\mathcal{S}'$ of $\mathcal{S}$ constructed by dividing each $S_\beta$ into  subsets  $S_{\beta,m}$ of size at most $d/2$. Importantly, from  \eqref{eq:upCARD} it follows that $|S_\beta|\leq (1+\delta) (1 + 1/\sqrt{C})^2 d$ and hence $S_\beta$ can be divided into at most $2(1+\delta) (1 + 1/\sqrt{C})^2$ parts of size at most $d/2$. Note that by applying Theorem \ref{th:OldProtocol} to $\mathcal{S}'$ and setting $C=1$ we get $q(\M,\mathcal{S}')\geq \frac{1}{8(1+\delta)} q(\M,\mathcal{S})\geq \frac{1}{32(1+\delta)}$. 

We can improve the greedy analysis of the sub-partition $\mathcal{S}'$ by noting that there cannot be too many large subsets $S_\beta$ and using inequality \eqref{eq:normbBoundKS} -- see Lemma \ref{lem:improvedBound} in Appendix \ref{app:aux} for details. Adapting the results presented therein, we can find a sub-partition $\mathcal{S}''$ whose elements have size at most $d/2$ and moreover $q(\M,\mathcal{S}'')\geq 0.068/(1+\delta)$, which is obtained by choosing $C=5$ and $\kappa=1/2$ in \eqref{eq:improvedBound}. 

We control the magnitude of $A_i$ as follows. From the definition of POVMs $\N^{(\beta)}$ in Eq. \eqref{eq:oldSIMPLeeffects} we get that for every $i$, $A_i=\alpha_i/\lambda_{\beta(i)}$, where $\beta(i)$ is the label of the unique subset $S_{\beta(i)}$ (of the partition $\mathcal{S}$) that contains $i$. From \eqref{eq:normbBoundKS} and by using assumptions about $\alpha_i$, we finally get $A_i\geq \frac{1}{1+\delta} \left( 1 + 1/\sqrt{C}  \right)^{-2}$. Additionally, since for every $\beta,m$ we have $\left\| \sum_{i\in{S_{\beta,m}}} M_i\right\| \leq  \left\| \sum_{i\in{S_{\beta}}} M_i\right\|$ we  still get that magnitudes $A_i$ of POVM elements $\N^{(\beta,m)}$  satisfy  $A_i\geq \frac{1}{1+\delta} \left( 1 + 1/\sqrt{C}  \right)^-2$ for any sub-partition of $\mathcal{S}$. Specifically, for a sub-partition $\mathcal{S}''$ constructed for $C=5$, we get $A_i \geq 0.47/(1+\delta)$. 
   
\end{proof}

By using similar reasoning as above we now prove a formal version of Result \ref{res:singleANC}, which gives bounds on success probability of simulation of POVMs as a function of the size of an ancilla which is allowed to be used. 

\begin{thm}[Formal version of Result \ref{res:singleANC} -- simulation via $k$-dimensional ancillas] \label{th:ancillas}Let $\M$ be a POVM on $\C^d$. Then there exists a probability distribution $\{p_\beta\}$ and a collection of POVMs $\{\N^{(\beta)}\}$ such that
\begin{itemize}
    \item[(i)] For every $\beta$ the POVM $\N^{(\beta)}$ can be implemented by a single projective measurement on $\C^d \ot \C^k$ (i.e. using ancilla of dimension $k$);
\item[(ii)] The convex combination $\L=\sum_{\beta}p_\beta \N^{(\beta)}$ simulates $\M$ with  postselection probability $q_k\geq1-\Theta(k^{-1/2})$,  i.e. $\sum_{\beta}p_\beta \N^{(\beta)}=(q_k\M,(1-q_k)\I_d)$. 
\end{itemize}
Additionally, for a one-qubit ancilla ($k=2$) we can find a simulation strategy which achieves $q_2\geq1/8=0.125$.
  
\end{thm}

\begin{proof}

We start with the proof for general $k$.  By repeating the reasoning from Section \ref{sec:overview} we can assume without loss of generality that $\M$ has rank-one effects $M_i=\alpha_i \ketbra{\psi_i}{\psi_i}$ ($i\in[n]$) and furthermore has ``flat'' effects: $\alpha_i\in[\tilde{\epsilon},\epsilon]$, for $\epsilon/\tilde{\epsilon}\leq 1+\delta$, where $\delta\in(0,1/2]$ is arbitrary and $\epsilon\leq \epsilon_\ast(\M,\delta)$ (at the end of the proof we will take the limit $\delta\rightarrow0$, $\epsilon\rightarrow 0$). From Theorem \ref{th:PostSimulation} we know that for any such POVM there exists a partition $\mathcal{S}=\{S_\beta\}_{\beta=1}^r$ of $[n]$ such that 
\begin{equation}\label{eq:repeatSIZE}
    q(\M,\mathcal{S})\geq (1+\sqrt{C})^{-2}\ \text{and}\ |S_\beta|\leq d \frac{\epsilon}{\Tilde{\epsilon}} (1 + 1/\sqrt{C})^2\ ,
\end{equation}
 where $C=r\epsilon$. Recall that $|S_\beta| +d - 1$ is the lower bound on the dimension of the total space ($d_{tot}$) needed to implement measurements $\N^{(\beta)}$ via Naimark extension. By setting $d_{tot}=d\cdot k$ and using Eq. \eqref{eq:repeatSIZE} we get 
 \begin{equation}\label{eq:ineqKversusC}
     \frac{\epsilon}{\Tilde{\epsilon}} (1 + 1/\sqrt{C})^2\leq k-1\ . 
 \end{equation}
The above inequality depends on the parameter $C=r\epsilon$, which can freely chosen (up to precision $\epsilon$, which can be set arbitrarily small). Furthermore, for $\epsilon/\tilde{\epsilon}\leq 3/2$ there always exists a satisfiable solution $C>0$ for every natural $k>2$. By solving Eq. \eqref{eq:ineqKversusC} for $\sqrt{C}$ and combing the result with the bound on $q(\M,\mathcal{S})$ realizable with chosen $k$ we obtain the bound on success probability attainable by a $k$-dimensional ancilla
\begin{equation}
    q(\M,\mathcal{S}) \geq \left(1-\left(\frac{\tilde{\epsilon}}{\epsilon}(k-1)\right)^{-1/2} \right)^2\ , 
\end{equation}
from which the main claim of the theorem follows (since $\tilde{\epsilon}/\epsilon =\Theta(1)$).

The case of $k=2$ requires more care. A constant lower bound on the success probability $q_2$ can be deduced by considering a greedy subpartition of $\mathcal{S}$ coming from Theorem \ref{th:PostSimulation}. However, a more thorough analysis given in Lemma \ref{lem:improvedBound} in Appendix \ref{app:aux} shows that many of the elements of the partition $\mathcal{S}$ will have smaller sizes than the one predicted by the bound \eqref{eq:repeatSIZE}. Specifically, taking \eqref{eq:improvedBound} for $C=1$ and $\kappa=1$ gives a subpartition $\mathcal{S}'$ such that $q(\M,\mathcal{S'})\geq [4 (1+\epsilon/\tilde{\epsilon})]^{-1}$, which in the limit\footnote{Formally, for every $\epsilon/\tilde{\epsilon}>1$ we have a simulation strategy which attains a slightly worse bound on the success probability than $1/8$. However, taking $\kappa\leq 1-1/d$ in \eqref{eq:improvedBound} and $\epsilon/\tilde{\epsilon}\leq 1/\kappa$ gives a simulation strategy utilizing a single auxiliary qubit and attaining $q(\M,\mathcal{S}')\geq 1/8$.} $\epsilon/\tilde{\epsilon}\rightarrow 1$ gives the claimed lower bound $q_2\geq 1/8$. 
\end{proof}

\section{Dimension-deficient Naimark theorem}\label{sec:defitientNAIMARK} 

In this section we formulate and prove dimension-deficient Naimark extension theorem. This result states that nearly projective measurements on $\C^d$ with small number of outcomes can be approximated (in a suitable sense) by projecively simulable measurements on $\C^d$. In what follows for a subspace $X\subset \C^d$ we denote by $|X|$ its dimension, by $\PP_X$ the orthogonal projector onto it, and by $X^\perp$ its orthogonal complement in $\C^d$. 

\begin{thm}[Dimension-deficient Naimark dilation]\label{th:dimentionDefitient}
    Let $l\leq d/2$ and let $\N=(N_1,\ldots,N_{l+1})$ be a POVM on $\C^d$ with effects satisfying $N_i=A_i \ketbra{\psi_i}{\psi_i}$, for $i\in[l]$. Let $W=\mathrm{span}_{\C}\lbrace\{ \ket{\psi_i} \rbrace\ ,\ i\in[l]\}$. Let $\F=(F_1,\ldots,F_{l+1})\in\POVM(\C^d)$ be defined by 
    \begin{equation}\label{eq:projSIMeffects}
        F_i= A_i \ketbra{\psi_i}{\psi_i} + \frac{1-A_i}{|W^\perp|} \PP_{W^\perp}\ ,
    \end{equation}
    for $i\in[l]$. Then the POVM $\F$ is projectively simulable, i.e. $\F\in\SP(d)$.
\end{thm}
\begin{proof}
We first note that Eq. \eqref{eq:projSIMeffects} uniquely defines a POVM $\F^\N$ on $\C^d$ because, due to normalization of POVMs, we have $F_{l+1}=\I_d - \sum_{i=1}^l F_i$ and consequently
\begin{equation}
   F_{l+1}=  \PP_W - \sum_{i=1}^l A_i \ketbra{\psi_i}{\psi_i} + \left(1- \frac{\sum_{i=1}^l (1-A_i)}{|W^\perp|}\right)\PP_{W^\perp}\ . 
\end{equation}
It is easy to see that $F_{l+1}\geq 0$. This is because (i) $\PP_W - \sum_{i=1}^l A_i \ketbra{\psi_i}{\psi_i} \geq 0$ (since $\ketbra{\psi_i}{\psi_i}$ have their support on the subspace $W\subset \C^d$) and (ii) due to the inequality $\sum_{i=1}^l (1- A_i)\leq d/2 \leq |W^\perp|$. We now proceed to show that $\F$ is projectively simulable. 

    Consider $\N^W\in\POVM(W)$ that is defined via $N^W_i= \PP_W N_i \PP_W $, for $i\in[l+1]$. We have $\sum_{i=1}^{l+1} \mathrm{rank}(N^W_i)\leq l+|W| \leq d$. It is therefore possible to realize the Naimark extension of $\N^W$ (cf. Theorem \ref{th:Naimark}) using the space $\C^d$. Let $\P^W=(P^W_1,\ldots, P^W_{l+1})$ be a projective measurement realizing the Naimark extension of $\N^W$. Because $\mathrm{rank}(N^W_i)=1$ for $i\in[l]$ we have $P^W_i = \ketbra{\phi_i}{\phi_i}$, where $\{\ket{\phi_i}\}_{i=1}^l$ is a collection of orthogonal vectors in $\C^d$. Because $\P^W$ is a Naimark extension of $\N^W$ we have for $i\in[l]$
    \begin{equation}\label{eq:purePURIFICATION}
        \PP_W \ketbra{\phi_i}{\phi_i}\PP_W = A_i \ketbra{\psi_i}{\psi_i}\ . 
    \end{equation}
    Furthermore, by using $\PP_W + \PP_{W^\perp}= \I_d$, $\tr(\ketbra{\psi_i}{\psi_i})=1$ and Eq. \eqref{eq:purePURIFICATION} we get that for $i\in[l]$
 \begin{equation}\label{eq:projectionsNORMALIZATIONS}
     \tr(\ketbra{\phi_i}{\phi_i} \PP_W) = A_i\ ,\ \tr(\ketbra{\phi_i}{\phi_i} \PP_{W^\perp}) = 1- A_i\ .
 \end{equation}
In order to get from the projective measurement $\P^W$ to $\F$ we show these POVMs can be connected by a suitably crafted unitary twirling
\begin{equation}\label{eq:twirlingPOVM}
    \F = \E_{U\sim \mathcal{E}} U \P^W U^\dag\ ,
\end{equation}
where $\mathcal{E}$ is the uniform measure on unitaries of the form $U=\exp(i \varphi) \PP_W \oplus V$, where $\varphi\in[0,2\pi]$ is a phase and $V\in U(W^\perp)$ is a unitary on $W^\perp$. The twirling in \eqref{eq:twirlingPOVM} is defined by applying averaging separately on every component of a POVM i.e. 
\begin{equation}
    \left[\E_{U\sim \mathcal{E}} U \P^W U^\dagger\right]_i = \E_{U\sim \mathcal{E}} U P^W_i U^\dagger\ .
\end{equation}
Using basic tricks from Haar measure integration it is easy to see that for $i\in[l]$
\begin{equation}\label{eq:intermediateidentity}
    \E_{U\sim \mathcal{E}} U P^W_i U^\dagger = A_i \psi_i + \frac{1-A_i}{|W^\perp|} \PP_{W^\perp} = F_i\ .
\end{equation}
Indeed, for every operator $B$ on $\C^d$ we have (cf. Lemma \ref{lem:auxINTEGR} in Appendix \ref{app:aux} for the proof)
\begin{equation}\label{eq:averaging}
    \E_{U\sim \mathcal{E}} U B U^\dagger = \PP_W B \PP_W + \frac{\tr(B \PP_{W^\perp})}{|W^\perp|} \PP_{W^\perp}\ .
\end{equation}
By inserting $B=P^W_i=\ketbra{\phi_i}{\phi_i}$ and using equations \eqref{eq:purePURIFICATION} and \eqref{eq:projectionsNORMALIZATIONS} we get \eqref{eq:intermediateidentity}. Finally, from the POVM normalization condition (unitary twirling is a unitary channel and therefore $\E_{U\sim \mathcal{E}} U \I_d U^\dagger= \I_d$) we get $\E_{U\sim \mathcal{E}} U P^W_{l+1} U^\dagger = \I_d - \sum_{i=1}^l F_i$. 

After establishing Eq. \eqref{eq:twirlingPOVM} we note that this identity proves $\F\in\SP(d)$. This is because unitary twirling can be interpreted as randomization (convex combination) over projective measurements $U \P^W U^\dagger$.

\end{proof}

\emph{Acknowledgements} We thank Adam Sawicki, Filip Maciejewski, Robert Huang, Yihui Quek, András Gilyén, Joao F. Doriguello, Ingo Roth,  Rafa\l\  Demkowicz-Dobrza\'{n}ski and  Marcin Kotowski for interesting discussions.  MK  acknowledges the financial support by  TEAM-NET project co-financed by EU within the Smart Growth Operational Programme (contract no.  POIR.04.04.00-00-17C1/18-00). MO ancnowledges the support of  National Science Centre, Poland under the grant OPUS: UMO2020/37/B/ST2/02478. Part of this work was conducted while MO was visiting the Simons Institute for the Theory
of Computing.

\bibliographystyle{apsrev4-2}
\bibliography{refs}

\appendix
\onecolumngrid

\section{Auxiliary technical results}\label{app:aux}
\noindent
In this part of the Appendix we collect a number of auxiliary technical results used throughout the paper.

\subsection{Flat fine-grainings of arbitrary measurements}
\noindent

The purpose of this section is to prove Lemma \ref{lem:flatPOVM}, which states that an arbitrary POVM $\M\in\POVM(\C^d)$ can be realized as coarse-graining of a POVM $\M'$  with rank-one effects that have \emph{nearly identical} traces. We shall need the following elementary lemma.

\begin{lem}\label{lem:devNUMBERS}
    Consider $x_1, \ldots, x_N > 0$. For any $\delta \in (0,1)$  there exists $\epsilon_\ast>0$ such that for all $\epsilon\in(0,\epsilon_\ast)$ there exists a subdivision $x_i = x_{i,1} + \ldots + x_{i,\ell_{i}}$, with $x_{i,j} > 0$ such that
   \begin{equation}
        \frac{\max\limits_{i,j} x_{i,j}}{\min\limits_{i,j} x_{i,j}} \leq 1+\delta\ ,
    \end{equation}
    and additionally for every $i,j$ we have $x_{i,j}\leq \epsilon$.
\end{lem}

\begin{proof}
    We use the fact that for any irrational number $\alpha$ the sequence $\{k\alpha\}_{k=1}^{\infty}$ is dense in $(0,1)$, where $\{\cdot\}$ denotes the fractional part of a real number. Therefore for $\Delta<1/2$ (its relation to $\delta$ will be explained later) for each irrational $x_i$ there exists $k_i$ such that $1-\Delta \leq \{k_i x_i\} < 1$. If $x_i$ is rational, we take $k_i$ to be any natural number such that $k_i x_i$ is integer. Thus if we divide each $x_i$ into parts of size $\frac{1}{k_i}$, with possibly one part of smaller size, we have
   \begin{equation}
    x_i = \frac{\lfloor k_i x_i\rfloor}{k_i} + \frac{\alpha_i}{k_i}
   \end{equation}
    with $1-\Delta\leq \alpha_i \leq 1$.
    
    Let now $k = k_1 \cdot \ldots \cdot k_N$. For each $i$ subdivide each part defined above further into $\frac{k}{k_i}$ parts of equal size. In this way we obtain for each $i$ parts of size $\frac{1}{k}$ plus possibly one part of size $\beta_i = \frac{\alpha_i}{k}$ satisfying
   \begin{equation}
    \frac{1-\Delta}{k} \leq \beta_i \leq \frac{1}{k}\ .
   \end{equation}
   The statement of the lemma follow by imposing the condition $\delta=1/(1-\Delta)-1=\Delta/(1-\Delta)$, setting $\epsilon_\ast =1/k$ and noting that the magnitude of each part can be further uniformly reduced by dividing each $\beta_i$ by an arbitrary large rational number.
\end{proof}

\newcounter{tempcounter}
\setcounter{tempcounter}{\value{lem}} 
\setcounter{lem}{0} 

\begin{lem}[Restatement]\  For every $\delta\in(0,1)$ and for every $\M\in\POVM(\C^d)$ there exists $\epsilon_\ast >0$ such that for all $\epsilon\in(0,\epsilon_\ast)$ there  exists a POVM $\M'$ and a stochastic map $\Q$ such that $\M=\Q(\M')$ and $M'_i=\alpha_i \ket{\psi_i}\bra{\psi_i}$, $\frac{\max_i \alpha_i}{\min_i \alpha_i} \leq 1 + \delta$, and furthermore $\max_i \alpha_i \leq \epsilon$.  
\end{lem}
\setcounter{lem}{\value{tempcounter}} 
\begin{proof}
    We start by diagonalizing the effects of the target POVM $\M=(M_1,\ldots,M_n)$,  $M_a= \sum_{j=1}^d \lambda_{j,a} \ketbra{\psi_{j,a}}{\psi_{j,a}}$. Clearly, the POVM $\M$ can be realized as coarse-graining of POVM $\tilde{\M}$ with effects $\tilde{M}_{j,a}=\lambda_{j,a} \ketbra{\psi_{j,a}}{\psi_{j,a}}$. Let $\mathcal{Q}_1$ be a stochastic map such that $\M=\mathcal{Q}_1(\tilde{\M})$. Let us now treat the positive number $\lambda_{j,a}$ as input to Lemma \ref{lem:devNUMBERS} (i.e. as numbers $x_1,\ldots, x_N$). It follows that for any $\delta\in(0,1)$ and any $\epsilon\leq \epsilon_\ast$ (with $\epsilon_\ast$ depending on the collection  $\{\lambda_{j,a}\}$) there exists a subdivision of numbers $\lambda_{j,a}$ into parts $\lambda_{j,a; m}$ (the range of $m$ can depend on $j$ and $a$) such that 
    \begin{equation}
       \frac{\max_{j,a,m}\lambda_{j,a; m}}{\min_{j,a,m}\lambda_{j,a; m}}\leq 1 +\delta \ \ \text{and}\ \  \max_{j,a,m}\lambda_{j,a; m} \leq \epsilon\ .
    \end{equation}
   A POVM $\M'$ with effects $\{\lambda_{j,a; m} \ketbra{\psi_{j,a}}{\psi_{j,a}} \} $ is a fine-graining of the POVM $\tilde{\M}$. Let $\mathcal{Q}_2$ be a a stochastic map such that $\tilde{M}=\mathcal{Q}_2(\M')$. We conclude the proof by noting that $\M=(\mathcal{Q}_1 \circ \mathcal{Q}_2 ) (\M')$.
\end{proof}

\setcounter{lem}{\value{tempcounter}} 

\subsection{Improved bounds for sizes of groups in the simulation alghorithm}

Here we give a detailed reasoning behind improved parameters of simulation of general POVMs by measurements with bounded number of outcomes (which are relevant for proofs of quantitative versions of Lemma \ref{lem:nearlyPROJpost} and Theorem \ref{th:ancillas}). The following Lemma can be regarded as a variant of Theorem \ref{th:PostSimulation} in the situation where we want to ensure that $|S_\beta|\leq \kappa d$.

\begin{lem}\label{lem:improvedBound}
Let $\M = (M_1,\ldots,M_n)$ be an $n$-outcome POVM on $\C^d$, with $M_i = \alpha_i \ketbra{\psi_i}{\psi_i}$, $\tilde{\epsilon}\leq \alpha_i \leq \epsilon$. Suppose that $\mathcal{S}=\lbrace S_\beta\rbrace_{\beta=1}^r$ is a partition of $[n]$ into disjoint subsets such that for all $\beta = 1, \ldots, r$ we have
\begin{equation}\label{eq:improvedBound-normbBoundKS}
        \left\| \sum_{i\in{S_\beta}} M_i \right\| \leq \frac{1}{r}\left( 1 + \sqrt{r\epsilon}  \right)^2.
\end{equation}
Let $C = r \epsilon$ and fix $\kappa > 0$. Then there exists a subpartition $\mathcal{S}' = \lbrace S'_\beta\rbrace_{\beta=1}^{r'} $ of $\mathcal{S}$ such that for each $\beta = 1, \ldots, r'$ we have $|S_{\beta}'| \leq \kappa d$ and
\begin{equation}\label{eq:improvedBound}
    q(\M,\mathcal{S}') \geq \frac{1}{(1 + \sqrt{C})^2 (\frac{\epsilon}{\tilde{\epsilon}}\cdot\frac{1}{\kappa C } + 1)},
\end{equation}
where $q(\M,\mathcal{S}')$ is defined in Eq.  \eqref{eq:explicitQsucc}.
\end{lem}
\begin{proof}
    For $j=0,1,2,\ldots$ let us say that a subset $S_\beta$ is of type $j$ if $j\cdot \kappa d\leq |S_\beta| < (j+1)\cdot \kappa d$. Let $L_j$ denote the set of indices of subsets of type $j$, i.e., $\beta \in L_j$ if $S_\beta$ is of type $j$. Since the total number of subsets is $r$, we have
    \begin{equation}\label{eq:improvedBound-sum-of-Li}
    \sum\limits_{j \geq 0} |L_j| = r.
    \end{equation}
    On the other hand, by counting the total number of elements of $[n]$ in groups of each type we have
    \begin{equation}\label{eq:improvedBound-weighted-sum-of-Li}
    \sum\limits_{j\geq 0} j\cdot \kappa d |L_j| \leq n.
    \end{equation}
    Let us divide each subset $S_\beta$ of type $j$ into $j+1$ subsets of size at most $\kappa d$ and denote the resulting partition by $\mathcal{S}' = \lbrace S_{\beta}' \rbrace_{\beta=1}^{r'} $. Sums of elements of the POVM contained in subsets $S_{\beta}'$ satisfy the same bound \eqref{eq:improvedBound-normbBoundKS} as for the original partition $\mathcal{S}$, i.e., we have for all $\beta = 1, \ldots, r'$
    \begin{equation}\label{eq:improvedBound-normbBoundKS-for-S'}
        \left\| \sum_{i\in S_{\beta}'} M_i \right\| \leq \frac{1}{r}\left( 1 + \sqrt{r\epsilon}  \right)^2.
    \end{equation}
As each subset $S_\beta$ of type $j$ splits into $j+1$ subsets, we obtain the bound
    \begin{eqnarray}\label{eq:improvedBound-normbBoundKS-sum}
        \sum_{\beta=1}^{r'} \left\| \sum_{i\in S_{\beta}'} M_i \right\| \leq \sum_{j \geq 0} \sum_{\beta \in L_j} (j+1) \frac{1}{r}\left( 1 + \sqrt{r\epsilon}  \right)^2 \nonumber \\ = \frac{1}{r}\left( 1 + \sqrt{r\epsilon}  \right)^2 \sum_{j \geq 0} (j+1)|L_j|.
    \end{eqnarray}
    Recalling the definition \eqref{eq:explicitQsucc} of $q(\M,\mathcal{S}')$, by employing \eqref{eq:improvedBound-sum-of-Li} and \eqref{eq:improvedBound-weighted-sum-of-Li} we obtain
    \begin{equation}\label{eq:improvedBound-qsucc}
        q(\M,\mathcal{S}') \geq \frac{1}{(1 + \sqrt{r\epsilon})^2 (\frac{1}{\kappa }\cdot\frac{n}{dr} + 1)}.
    \end{equation}
    Since $(M_1, \ldots, M_n)$ is a POVM on $\C^d$ satisfying $\tr M_i \geq \tilde{\epsilon}$, we have $n\tilde{\epsilon} \leq d$, which translates to
    \begin{equation}\label{eq:improvedBound-qsucc-with-eps}
        q(\M,\mathcal{S}') \geq \frac{1}{(1 + \sqrt{C})^2 (\frac{\epsilon}{\kappa  C \tilde{\epsilon}} + 1)}\ .
    \end{equation}
\end{proof}

\subsection{Integration formula}
\noindent

In this section we provide a proof of Eq. \eqref{eq:averaging}, which is the missing step in the proof of Theorem \ref{th:dimentionDefitient}.

\begin{lem}\label{lem:auxINTEGR}
Let $\mathcal{E}$ be the ensemble of unitary matrices $U$ on $\C^d$ such that $U=\exp(i\varphi)\PP_W\oplus V$, where $\varphi$ is uniformly distributed on $[0,2\pi]$, $V$ is distributed according to the Haar measure on $U(W^\perp)$ and $V$, $\varphi$ are independent. Let $B$ be a linear operator on $\C^d$. Then we have
\begin{equation}
    \E_{U\sim \mathcal{E}} U B U^\dagger = \PP_W B \PP_W + \frac{\tr(B \PP_{W^\perp})}{|W^\perp|} \PP_{W^\perp}\ .
\end{equation}
\end{lem}
\begin{proof}
    We start with a decomposition
    \begin{equation}
        B=(\PP_W + \PP_{W^\perp}) B (\PP_W + \PP_{W^\perp}) = \PP_W B\  \PP_W +\PP_{W^\perp} B\ \PP_{W^\perp} +
        \PP_{W} B\ \PP_{W^\perp} +\PP_{W^\perp} B\ \PP_{W}\ .
    \end{equation}
    After conjugating $B$ by $U=\exp(i\varphi)\PP_W\oplus V$, where $V$ is supported on $W^\perp$, we get
    \begin{equation}
      U B U^\dagger =  \PP_W B\  \PP_W +V \PP_{W^\perp} B\ \PP_{W^\perp} V^{\dagger} +
        \exp(i \varphi)\PP_{W} B\ \PP_{W^\perp} V^\dagger + V \PP_{W^\perp} B\ \PP_{W} \exp(- i \varphi)\ . 
    \end{equation}
    Since $\varphi$ is distributed uniformly on $[0,2\pi]$ the terms involving $\exp(\pm i\varphi)$ average to $0$ and consequently
    \begin{equation}
          \E_{U\sim \mathcal{E}} U B U^\dagger = \PP_W B \PP_W + \E_{V\sim \mu(W^\perp)} V \PP_{W^\perp} B\ \PP_{W^\perp} V^{\dagger}\ ,
    \end{equation}
    where $\mu(W^\perp)$ denotes the Haar measure on the unitary group $U(W^\perp)$. To conclude the proof, note that $\PP_{W^\perp} B\ \PP_{W^\perp}$ is a linear operator on $W^\perp$ and from the properties of Haar measure its averaged version $\E_{V\sim \mu(W^\perp)} V \PP_{W^\perp} B\ \PP_{W^\perp} V^{\dagger}$  must equal $\lambda \PP_{W^\perp}$ ($\PP_{W^\perp}$ acts as identity on subspace $W^\perp$).     
    The proportionally constant $\lambda$ equals $\tr(B \PP_{W^\perp})/ |W^\perp|$, which follows form the fact that the map $X\mapsto \E_{V\sim \mu(W^\perp)} V X V^{\dagger}$ is trace preserving for operators $X$ supported on $W^\perp$.
\end{proof}

\subsection{From dimension-deficient Naimark to simulation under depolarizing noise}\label{app:dilationTOwhitenoise}

For $l\leq d/2$ the projectively simulable POVM $\F$ from Dimension-Deficient Naimark theorem (Theorem \ref{th:dimentionDefitient}) realizes perfectly a measurement of the form 
\begin{equation}
\N=( A_1 \ket{\psi_1}\bra{\psi_1},\ldots, A_l \ket{\psi_l}\bra{\psi_l}, \I_d - \sum_{j=1}^l A_j \ket{\psi_j}\bra{\psi_j})\ , \ \   
\end{equation}
on states supported on the subspace $W=\mathrm{span}_{\C}\{\ket{\psi_i}\ |\ i\in[l]\}$. The following Lemma, stated previously as Lemma \ref{lem:simnearlyPROJ},  shows that a slight modification of $\F$ gives a projective simulation of $\mathrm{\mathrm{\Phi}}_t(\N)$ for $t=\Theta(\min_{i\in[l]} A_i)$.

\begin{lem}[Simulation of nearly projective measurements under depolarizing noise -- full version of Lemma \ref{lem:simnearlyPROJ}]\label{lem:finPROJsimulation}
    Let $l\leq d/2$ and let $\N=(N_1,\ldots,N_{l+1})\in\POVM(\C^d)$ be a POVM of the form 
    \begin{equation}
         N_i = 
       \begin{cases}
        A_i \ket{\psi_i}\bra{\psi_i} & \text{if}\ i\in [l]  \\
        \I_d -  \sum_{j=1}^l A_j \ket{\psi_j}\bra{\psi_j}& \text{if}\   i=l+1 \\
        \end{cases}\ .
    \end{equation}
    Then $\mathrm{\mathrm{\Phi}}_{t}(\N)\in\SP(d)$ for
    \begin{equation}
    t\leq t_\N=\min_{i\in[l]} \frac{|W^\perp| A_i}{|W|(1-A_i)+|W^\perp|}\ ,
\end{equation}
where $W=\mathrm{span}_{\C}\{\ket{\psi_i}\ |\ i\in[l]\}$, $W^\perp$ is the orthogonal complement of $W$ in $\C^d$, and $|X|$ denotes the dimension of linear subspace $X\subset \C^d$. We remark that since $l \leq d/2$, we have $|W|\leq |W^\perp|$ and thus $t_\N=\Theta( \min_{i\in[l]} A_i)$.
\end{lem}

\begin{proof}

Recall that the projectively simulable $l+1$-outcome measurement  $\F$ from Theorem \ref{th:dimentionDefitient} has effects of the form
\begin{equation}\label{eq:psSUB}
     F_i= A_i \ketbra{\psi_i}{\psi_i} + \frac{1-A_i}{|W^\perp|} \PP_{W^\perp}\ ,
\end{equation}
for $i\in[l]$. On the other hand, the dephased version of $\N$ has effects
\begin{equation}\label{eq:noisyPOVM}
\mathrm{\Phi}_\tau(N_i)= \tau\cdot A_i \ketbra{\psi_i}{\psi_i} + (1-\tau) \cdot \frac{A_i}{d}\left(\PP_W + \PP_{W^\perp}\right)\ ,
\end{equation}
for $i\in[l]$. Let $\mathbf{C}$ be a $l+1$-outcome POVM with effects satisfying
\begin{equation}\label{eq:classPOVM}
    C_i = a_i \PP_W + b_i \PP_{W^\perp}\ ,
\end{equation}
for $a_i, b_i\geq 0$ and $i\in[l]$. Clearly, $\mathbf{C}$ defined in this way is projectively simulable\footnote{For example because $\mathbf{C}$ can be realized by a post-processing of a dichotomic measurement $(\PP_W,\PP_{W^{\perp}})$}. The form of equations describing $F_i$,$\mathrm{\Phi}_\tau(N_i)$ and $C_i$ suggest to consider $\L= \tau \F +(1-\tau) \mathbf{C}$, or appropriate choice of $a_i$, $b_i$, as a projectively simulable measurement realizing $\mathrm{\Phi}_\tau(\N)$. Imposing  $ \tau\ F_i +(1-\tau)\ C_i = \mathrm{\Phi}_\tau(N_i)$ for $i\in[l]$ yields\footnote{The condition $\tau F_{l+1} +(1-\tau) C_{l+1} = \mathrm{\Phi}_\tau(N_{l+1})$ is automatically satisfied provided the first $l$ equations hold.} the condition
\begin{equation}
    \tau\cdot \frac{1-A_i}{|W^\perp|} \PP_{W^\perp} + (1-\tau)\cdot \left(a_i \PP_W + b_i \PP_{W^\perp} \right)= (1-\tau)\cdot \frac{A_i}{|W|+|W^\perp|} \left( \PP_W + \PP_{W^\perp} \right)\ .
\end{equation}
It follows that for $i\in[l]$
\begin{equation}
    a_i = \frac{A_i}{|W|+|W^\perp|}\ \ ,\ \ b_i = \frac{A_i}{|W|+|W^\perp|}+\frac{\tau}{1-\tau} \frac{A_i-1}{|W^\perp|}\ .
\end{equation}
The above coefficients have to satisfy $a_i,b_i\geq 0$ and
\begin{equation}
    \sum_{i=1}^l a_i \leq 1\ \ ,\ \ \sum_{i=1}^l b_i \leq 1
\end{equation}
for $\mathbf{C}$ to be a valid POVM. It can be shown that the largest $\tau$ for which above equations hold equals
\begin{equation}
    \tau_\ast=t_\N= \min_{i\in[l]} \frac{|W^\perp| A_i}{|W|(1-A_i)+|W^\perp|}\ .
\end{equation}
The fact that $\mathrm{\Phi}_t(\N)\in\SP(d)$ for $t\leq t_\N$ follows either from realizing the above simulation strategy for  $\tau\leq t_\N$ or from simply mixing $\mathrm{\Phi}_{t_\N}(\N)$ with $\mathrm{\Phi}_0(\N)\in\SP(d)$ with appropriate weights. 
\end{proof}

\section{Proofs of statements concerning applications of main results}\label{app:Applications}

In this part we present missing proofs of Propositions given in Section \ref{sec:mainRES}.

\setcounter{tempcounter}{\value{prop}} 
\setcounter{prop}{1} 

\begin{prop}[Restatement]
Let $\mathcal{O}=\lbrace O_i\rbrace_{i=1}^L$ be a collection of observables on $\C^d$ satisfying $\tr(O_i)=0$, for $i\in[L]$. Let $\M=(M_1,\ldots,M_n)$ be a POVM that can be used to estimate expectation values of observables $O\in\mathcal{O}$. Let $\hat{e}_O$ be an unbiased estimator of the expectation value of $O$, i.e. a real-valued function $\hat{e}_O:[n]\rightarrow \mathbb{R}$ satisfying
\begin{equation}\label{eq:unbiased}
    \mathbb{E}\hat{e}_O = \sum_{i=1}^n \hat{e}_O(i) \tr(\rho M_i) = \tr(\rho O)\ ,
\end{equation}
for every state $\rho$. Let $\Delta_\M(O,\rho)=\mathbb{E}\hat{e}_{O}^2$ be the upper bound on the variance of $\hat{e}_O$. Then, for $c=0.02$  from Result \ref{res:projSIM}, a projectively simulable POVM $\N=\mathrm{\Phi}_c (\M)\in \SP(d)$ 
can be used to estimate expectation values of observables $O\in\mathcal{O}$ via estimators $\hat{e}'_{O}(k)\coloneqq\frac{1}{c} \hat{e}_O(k) $. Furthermore we have $ \max_\rho \Delta_\N(O,\rho)\leq \frac{1} {c^2} \max_\rho\Delta_\M(O,\rho)$.
\end{prop}
\setcounter{prop}{\value{tempcounter}} 
\begin{proof}
    We first prove that $\hat{e}'_O(i)=(1/c)\hat{e}_O(i)$ is an unbiased estimator of expectation value of $O$ once the measurement outcomes are collected via POVM $\N=\mathrm{\Phi}_c(\M)$: 
\begin{equation}
     \mathbb{E}\hat{e}'_O  =   \sum_{i=1}^n \hat{e}'_O(i)  \tr(\rho \mathrm{\Phi}_c(M_i)) = 
        \frac{1}{c}\sum_{i=1}^n \hat{e}_O(i) \tr(\mathrm{\Phi}_c(\rho) M_i) = \frac{1}{c} \tr(\mathrm{\Phi}_c(\rho) O)= \tr(\rho O)\ ,
     \end{equation}
     where in the second equality we used $\tr(\rho\mathrm{\Phi}_c(M_i))=\tr(\mathrm{\Phi}_c(\rho)M_i)$, in the  third equality we used \eqref{eq:unbiased}, and in the fourth equality we used the definition of depolarizing channel and $\tr(O)=0$.

     The proof that $ \max_\rho \Delta_\N(O,\rho)\leq \frac{1} {c^2} \max_\rho\Delta_\M(O,\rho)$ is similar:

\begin{equation}
    \Delta_\N(O,\rho)  =   \sum_{i=1}^n \hat{e}'_O(i)^2  \tr(\rho \mathrm{\Phi}_c(M_i)) = 
        \frac{1}{c^2}\sum_{i=1}^n \hat{e}_O(i)^2 \tr(\mathrm{\Phi}_c(\rho) M_i) =\frac{1}{c^2}\left(c \Delta_\M(O,\rho) +(1-c) \Delta_\M(O,\frac{\I_d}{d}) \right)\ . 
     \end{equation}
 Finally, we obtain the desired inequality by observing  $\Delta_\M(O,\frac{\I_d}{d}) )\leq \max_\rho \Delta_\M(O,\rho)$, inserting it into above equality and again optimizing both sides over $\rho$.
\end{proof}

\setcounter{tempcounter}{\value{prop}} 
\setcounter{prop}{3} 

\begin{prop}[Restatement]
     Let $c=0.02$ be the constant appearing in Result \ref{res:projSIM}. Then
 \begin{itemize}
     \item[(i)] States $\rho_{iso}(t)$ are POVM-local for $t\leq c\ t^{PM}_{iso}\approx \frac{c \log(d)}{d}$.
    \item[(ii)] For any pure state $\ket{\psi}$ on $\C^d \ot \C^d$ states $\rho_{\psi}(t)$ are POVM-local for $t\leq  \frac{c\ t^{PM}_{iso}}{(1-c\ t^{PM}_{iso})(d-1)+1} \approx \frac{c \log(d)}{d^2}$.
 \end{itemize}
\end{prop}
\setcounter{prop}{\value{tempcounter}}
\begin{proof}
    We follow a general strategy of using local models for projective measurements to construct local models for POVMs that work for  more noisy states (following the general approach outlined in \cite{Bowels2015,Oszmaniec17,Hirsch2017betterlocalhidden}). 
    
    Specifically, assume that we have a bipartite state $\rho$ on $\H_A \otimes \H_B$ which is local for all POVMs $\M\in\POVM(\H_A)$ on Alice's side and all projective measurements $\P\in\PP(\H_B)$ on Bob's side. Let $c$ be such that $\mathrm{\Phi}_c(\N)\in\SP(\H_B)$ for all $\N\in\POVM(\H_B)$. Then the state\footnote{We denote by $\mathrm{id}$ the identity channel on Alice side.} $(\mathrm{id}{\ot}\mathrm{\Phi}_c) (\rho)$ is POVM-local. To realize this, observe that for all quantum measurements $\M\in\POVM(\H_A)$, $\N\in\POVM(\H_B)$ and all outcomes $a,b$ we have
\begin{equation}
    \tr\left((\mathrm{id}{\ot}\mathrm{\Phi}_c)(\rho) M_a \otimes N_b\right) =   \tr\left( \rho M_a \otimes \mathrm{\Phi}_c(N_b)\right) = \sum_\beta p_\beta \tr\left( \rho M_a \otimes P^{(\beta)}_b\right) \ ,
\end{equation}
where in the second equality we used $\mathrm{\Phi}_c(\N)\in\SP(\H_B)$ to decompose $\mathrm{\Phi}_c(\N)$ as a convex combination of projective measurements  $\mathrm{\Phi}_c(\N)=\sum_\beta p_\beta \P^{(\beta)}$ (note that both $\{p_\beta\}$ and $\{\P^{(\beta)}\}$ in general depend on $\N$). 
Consequently, arbitrary  correlations on the state $(\mathrm{id}{\ot}\mathrm{\Phi}_c)(\rho)$ can be explained by a local model -- they can be decomposed as a convex mixture of correlations obtained on $\rho$ where Alice performs a general POVM and Bob performs a projective measurement, but for these measurements $\rho$ is already local. 

To prove (i) we recall that the models for projective measurements developed for $\rho_{iso}(t)$ in \cite{Almeida2007,Wiseman2007} have the desired property  -- they in fact allow arbitrary POVMs on Alice side while restricting to projective measurements on Bob side. Specifically, the  hidden variable space $\Lambda$ considered therein  coincides with the space of pure states on $\C^d$ and Alice's response function is given by \begin{equation}\label{eq:responseAlmeida}
    \xi_A(a|\M,\lambda)=\tr(M_{a}^T \ketbra{\lambda}{\lambda})\ . 
\end{equation}
 Because of this we can apply the above logic to claim that for every $t\leq t^{PM}_{iso}$ the state $(\mathrm{id}\otimes \mathrm{\Phi}_c) (\rho_{iso}(t))$ is POVM-local. The claim (i) follows from the simple identity $(\mathrm{id}\otimes \mathrm{\Phi}_c) (\rho_{iso}(t))=\rho_{iso}(c t)$.

The proof of (ii) follows the analogous steps as the construction from \cite{Almeida2007}. Therein, the authors adapted Nielsen's deterministic LOCC conversion protocol (that  deterministically maps $\ket{\phi_d}$ into an arbitrary bipartite state $\ket{\psi}$ on $\C^d\otimes \C^d$) to hidden variable models for $\rho_{iso}(t)$ such that:
\begin{itemize}
    \item[(a)] They have ``quantum mechanical'' expectation response function on Alice side (cf. \eqref{eq:responseAlmeida});
\item[(b)] The set of measurements $\N$ on Bob's side for which the model works (denoted by $\mathcal{B}$) is invariant under local unitary operations applied on his part of the system. 
\end{itemize}

\noindent
The net result of the analysis from \cite{Almeida2007} is the following implication: if $\rho_{iso}(\tau)$ has a local hidden variable model satisfying (a) and (b), then for every state $\rho_\psi(t)$ with 
\begin{equation}\label{eq:Nielsbound}
    t\leq  \frac{\tau}{(1-\tau)(d-1)+1}
\end{equation}
there exists a local model valid for all POVMs on Alice's side and measurements $\N\in\mathcal{B}$ on Bob's side.   Clearly, the model constructed in  (i) satisfies condition (a) and (b) for $\mathcal{B}=\POVM(\H_B)$. Therefore we can conclude (ii) from \eqref{eq:Nielsbound} by  $\tau = c t^{PM}_{iso}$.

\end{proof}

\section{Random partitions give quite good simulation via postselection}\label{app:randomPART}

In this part we prove that by randomly choosing the partition $\mathcal{S}=\{S_\beta\}_{\beta=1}^r$, in conjunction with the simulation protocol from Theorem \ref{th:OldProtocol}, we can simulate arbitrary POVMs on $\C^d$ by POVMs with $\Theta(d)$ outcomes and success probability $q$ scaling like $1/\log(d)$. Due to the structure of the simulation protocol we also get that a random partition allows to simulate an arbitrary POVM on $\C^d$ by POVMs requiring only a single auxiliary qubit, with success probability still at least $\Theta(1/\log(d))$.

We start with a technical result which ensures that for every extremal POVM $\M$ on $\C^d$ with rank-one effects one can define a \emph{fine-grained} version of it, $\M'$, that has still $O(d^2)$ effects but their operator norm is at most $\frac{1}{d}$.

\begin{lem}\label{lem:fineGrainingExtremal} Let $\M=(M_1,\ldots, M_n)\in\POVM(\C^d)$ be an extremal POVM on $\C^d$ with rank-one effects. Then there exists a POVM $\M'$ and a stochastic map $\Q$ such that $\M=\Q(\M')$ and the following holds:
\begin{itemize}
        \item[(i)] The number of outcomes $n'$ of $\M'$ satisfies $ n'\leq 2d^2$.
    \item[(ii)] For every $j=1.\ldots, n'$ we have $M'_j = \alpha'_j \ketbra{\psi_j}{\psi_j}$, with $\alpha_j\leq \frac{1}{d}$. 
\end{itemize}
\end{lem}

\begin{proof}
 First we note that every extremal POVM on $\C^d$ with rank-one effects has at most $n=d^2$ nonzero effects $M_i= \alpha_i \ketbra{\psi_i}{\psi_i}$  (cf. \cite{DAriano2005}). The proof  is analogous to that of Lemma \ref{lem:flatPOVM} and amounts to splitting every $\alpha_i$ (and the corresponding outcome) into $\ceil{d\alpha_i}$ parts so that the resulting $\alpha'_j$ have each magnitude smaller than $\frac{1}{d}$. 
 
 We now count by how much the number of outcomes can grow in the course of the above process. To this end we define for $a=0,1,2,\ldots$
 \begin{equation}
     I_a\coloneqq \left\{i\ \middle|  \ \frac{a-1}{d}\leq \alpha_i < \frac{a}{d} \right\}\ .
 \end{equation}
It is easy to observe that the number of outcomes after the division process equals (compare the similar reasoning in the proof of Lemma \ref{lem:improvedBound}):
\begin{equation}\label{eq:totSUMprim}
    n'= \sum_{a} a |I_a|\ , 
\end{equation}
where $|I_a|$ denotes the cardinality of $I_a$.  We also have 
\begin{equation} \label{eq:totSUM}
    n= \sum_a |I_a|\ ,
\end{equation}
and furthermore
\begin{eqnarray}\label{eq:finSIMPLE}
    d= \sum_{i=1}^n \alpha_i = \sum_a \sum_{i\in I_a} \alpha_i \geq \sum_a \sum_{i\in I_a}  \frac{a-1}{d} = \frac{1}{d} \left(   \sum_a a |I_a| - \sum_a |I_a| \right) .
\end{eqnarray}
Combining Equations \eqref{eq:totSUM}, \eqref{eq:finSIMPLE} and using the bound $n\leq d^2$, we obtain
\begin{equation}
    n'\leq d^2 +n \leq 2d^2\ .
\end{equation}

\end{proof}

The following theorem shows that for fine-grainings $\M'$ considered in the lemma above a random choice of partition $\mathcal{S}=\{S_\beta\}_{\beta=1}^r$ of size $r=Cd$ gives  success probability scaling like $1/\log(d)$ while size of each $|S_\beta|$ scales linearly with $d$.

\begin{thm}\label{th:PostSimulation-randomized}
    Let $\M = (M_1, \ldots, M_n)$ be a POVM on $\C^d$ with rank-one effects $M_i = \alpha_i \ketbra{\psi_i}{\psi_i}$ such that $n \leq 2d^2$ and $\alpha_i \leq \frac{1}{d}$. Fix $r = Cd$, $C>0$, and consider a random partition $\mathcal{S} = \{ S_\beta \}_{\beta=1}^{r}$ of $[n]$ into disjoint subsets obtained in the following way -- each element $i \in [n]$ is assigned to a subset $S_\beta$ chosen uniformly at random (with probability $\frac{1}{r}$). Then with probability at least $1/4$  (for fixed POVM $\M$, over the choice of the random partition) the following holds:
    \begin{itemize}
        \item[(i)]
        We have \begin{equation}\label{eq:KSqsucc-randomized}
    q(\M,\mathcal{S})\geq  \frac{1}{3.44 + 2  C \log d}\ , 
        \end{equation}
        where $q(\M,\mathcal{S})$ is the success probability of the simulation protocol from Theorem \ref{th:OldProtocol}, given by Eq. \eqref{eq:explicitQsucc}.
        \item[(ii)] For all $\beta = 1, \ldots ,r$ we have
        \begin{equation}\label{eq:upCARD-randomized}
    |S_{\beta}| \leq \frac{2}{C}(1 + \delta) d\ ,
\end{equation}
where $\delta$ satisfies $\delta^3 \geq \frac{3C}{2d} \left( 1+ \log(4Cd) \right)$.
    \end{itemize}
\end{thm}

\
\begin{proof}

Consider random matrices $X_i$, $i=1, \ldots,n$, sampled independently in the following way -- we take
\[
X_i = \I_d \otimes \ldots \otimes M_i \otimes \ldots \otimes \I_d,
\]
where the tensor product is $r$-fold and $M_i$ occurs at the $\beta$-th component, $\beta=1, \ldots, r$, with probability $\frac{1}{r}$. 
The random partition $\mathcal{S} = \{ S_\beta \}_{\beta=1}^{r}$ is defined by taking $i \in S_\beta$ if $X_i$ is nontrivial on the $\beta$-th component of the tensor product.

Let $Y = \sum\limits_{i=1}^{n} X_i$. Because of the tensor product form of each $X_i$ we have
\begin{equation}\label{eq:norm-y}
\norm{Y} = \sum\limits_{\beta=1}^{r} \left\| \sum\limits_{i \in S_\beta} M_{i} \right\|.
\end{equation}
We compute the expectation of $Y$:
\begin{equation}
    \E Y = \sum\limits_{i=1}^{n} \E X_i = \frac{1}{r} \sum\limits_{i=1}^{n} \left( M_i \otimes \ldots \otimes \I_d + \ldots + \I_d \otimes \ldots \otimes M_i \right) = \I_d \otimes \ldots \otimes \I_d\ ,
\end{equation}
where in the last equality we used that $M_i$ sum up to $\I_d$ on each coordinate of the tensor product.
By the assumption $\| M_i \| \leq \frac{1}{d}$ and multiplicativity of the operator norm for tensor products we have $\| X_i \| \leq \frac{1}{d}$ for each $i \in [n]$. We can thus employ \cite[(5.1.8)]{tropp} (with $L = \frac{1}{d}$, $\mu_{max}=1$ and substituting the dimension of the tensor product space $d^r$ for $d$), obtaining
\begin{equation}
\E \| Y \| \leq 1.72 + \frac{1}{d} \log d^r = 1.72 + C \log d\ .
\end{equation}
An application of Markov's inequality (to the random variable $\|Y\|$) gives that with probability at least $1/2$ we have $\| Y \| \leq 3.44 + 2C \log d$. Recalling Eq. \eqref{eq:explicitQsucc}, we obtain that the success probability of the protocol from Theorem \ref{th:OldProtocol} satisfies 
\begin{equation}\label{eq:randomized-success-prob}
    q(\M, \mathcal{S}) \geq \frac{1}{\| Y \|} \geq \frac{1}{3.44 + 2C \log d}
\end{equation}
with probability at least $1/2$ over the choice of $\mathcal{S}$.

To obtain the bound on the group sizes $|S_\beta|$, note that their joint distribution is multinomial with parameters $n$ and $(\frac{1}{r}, \ldots, \frac{1}{r})$. Thus $|S_\beta| = \frac{n}{r}$ for any $\beta=1,\ldots,r$ and a standard Chernoff bound shows that
\begin{equation}
    \Pp\left( |S_\beta| \geq (1+\delta) \frac{n}{r} \right) \leq e^{-\frac{\delta^2 n}{3r}} .
\end{equation}
A union bound over $\beta$ shows that $\max\limits_{\beta =1,\ldots, r} |S_\beta| \leq (1+\delta) \frac{n}{r}$ with probability at least $1-re^{-\frac{\delta^2 n}{3r}}$. Without loss of generality we can assume that actually $n=2d^2$. After inserting $r=Cd$ we obtain that for fixed $\delta$ satisfying, say, $\delta^3 \geq \frac{3C}{2d} \left( 1+ \log(4Cd) \right)$, with probability at least $3/4$ all groups satisfy $|S_\beta| \leq \frac{2}{C}(1 + \delta) d$.

Together with Eq. \eqref{eq:randomized-success-prob} this implies that the partition $\mathcal{S}$ satisfies both properties from the statement of the theorem with probability at least $1/4$.
\end{proof}

    We remark that in the above theorem we did not attempt to optimize the involved parameters so as to optimize $q(\M,\mathcal{S})$ and $\max_\beta |S_\beta|$.     The qualitative conclusion of Theorem \ref{th:PostSimulation-randomized} is the following.

    \begin{corr}\label{corLefficeint}
        Let $\M\in\POVM(\C^d)$ be a POVM on $\C^d$ satisfing assumptions of Theorem \ref{th:PostSimulation-randomized}. Then there exist a randomized algorithm running in time $\mathrm{poly}(d)$ which 
        \begin{itemize}
            \item[(i)] Finds a partition $\mathcal{S}=\{S_\beta\}_{\beta=1}^r$ such that success probability of protocol from Theorem \ref{th:OldProtocol} satisfies $q(\M,\mathcal{S})\geq \Theta(1/\log(d))$;
            \item[(ii)] Finds a description of POVMs $\N^{(\beta)}$ (used for the simulation of $\M$) with at most $\Theta(d)$ effects (out of which all but one have rank one).
        \end{itemize}  By further spitting subsets $S_\beta$ into sets of size at most $d$ (which can be done in time $\Theta(d)$), we obtain a partition  $\mathcal{S}'=\{S'_\gamma\}_{\gamma=1}^{r'}$  such that the success probability $q(\M,\mathcal{S}')\geq \Theta(1/\log(d))$ and each POVM $\N^{(\gamma)}$ (corresponding to an element of the subpartition $S'_\gamma$) can be realized by just a single auxiliary qubit. 
    \end{corr}

\end{document}